\documentclass[twocolumn,twoside]{IEEEtran}
\usepackage{times} 
\usepackage{epsfig}
\usepackage{graphicx}
\usepackage{epstopdf}
\usepackage{color,subfigure, latexsym, amsmath, amsfonts, amssymb,cite}
\usepackage{algorithm}
\usepackage{algorithmic}
\usepackage{epsfig}
\usepackage{graphicx}
\usepackage{epstopdf}
\usepackage{bm,amsmath,amssymb,amsfonts,graphicx,epsfig,amsthm,color, xcolor}
\usepackage{bbold,dsfont}
\usepackage{float}

\usepackage[hyphens]{url}
\PassOptionsToPackage{bookmarks=false}{hyperref}

\usepackage[depth=-1]{bookmark}

\usepackage{booktabs}       
\usepackage{amsfonts}       
\usepackage{microtype}      

\usepackage{times}
\usepackage{graphicx} 

\usepackage{wrapfig}



\usepackage{accents}
\newtheorem{proposition}{Proposition}
\newtheorem{lemma}{Lemma}

\newtheorem{theorem}{Theorem}
\newtheorem{definition}{Definition}
\newtheorem{assumption}{Assumption}
\theoremstyle{remark}\newtheorem{remark}{Remark}

\providecommand{\keywords}[1]{\textbf{\textit{Index terms---}} #1}

\allowdisplaybreaks
\title{\LARGE \bf
	Data-Driven Resilient  Predictive Control under Denial-of-Service
}
\author{Wenjie~Liu,~Jian~Sun,~\IEEEmembership{Senior Member,~IEEE}, Gang Wang,~\IEEEmembership{Member,~IEEE},\\Francesco Bullo,~\IEEEmembership{Fellow,~IEEE},~and
	Jie Chen,~\IEEEmembership{Fellow,~IEEE}
	\thanks{This work was supported in part by the National Key R\&D Program of China under Grant 2018YFB1700100, and the National Natural Science Foundation of China under Grants 61925303, 62088101, U20B2073, 61720106011, 62173034. 
	}
	\thanks{W. Liu and G. Wang are with the State Key Lab of Intelligent Control and Decision of Complex Systems and the School of Automation, Beijing Institute of Technology, Beijing 100081, China (e-mail: liuwenjie@bit.edu.cn; gangwang@bit.edu.cn).
		
		J. Sun is with the State Key Lab of Intelligent Control and Decision of Complex Systems and the School of Automation, Beijing Institute of Technology, Beijing 100081, China, and the Beijing Institute of Technology Chongqing Innovation Center, Chongqing 401120, China (e-mail: sunjian@bit.edu.cn).
		
		F. Bullo is with the Mechanical Engineering Department and the Center of Control, Dynamical Systems and Computation, UC Santa Barbara, CA 93106-5070, USA (e-mail: bullo@ucsb.edu).
		
		J. Chen is with the Department of Control Science and Engineering, Tongji University, Shanghai 201804, China, and also with the State Key Lab of Intelligent Control and Decision of Complex Systems and the School of Automation, Beijing Institute of Technology, Beijing 100081, China 	
		(e-mail: chenjie@bit.edu.cn).
	}
}
\begin{document}
	\maketitle

	\begin{abstract}
		The study of resilient control of linear time-invariant (LTI) systems
		against denial-of-service (DoS) attacks is gaining popularity in emerging
		cyber-physical applications.  In previous works, explicit system models
		are required to design a predictor-based resilient controller.  These
		models can be either given \emph{a priori} or obtained through a prior
		system identification step.  Recent research efforts have focused on
		data-driven control based on pre-collected input-output trajectories
		(i.e., without explicit system models).  In this paper, we take an
		initial step toward data-driven stabilization of stochastic LTI systems
		under DoS attacks, and develop a resilient model predictive control (MPC)
		scheme driven purely by data-dependent conditions.  The proposed
		data-driven control method achieves the same level of resilience as the
		model-based control method. For example, local input-to-state
		stability (ISS) is achieved under mild assumptions on the noise and the
		DoS attacks.  To recover global ISS, two modifications are further
		suggested at the price of reduced resilience against DoS attacks or
		increased computational complexity.  Finally, a numerical example is
		given to validate the effectiveness of the proposed control method.

	\end{abstract}
	
	\section{Introduction}\label{sec:intro}
	Thanks to recent advances in computing and networking technologies, recent
	years have witnessed rapid developments in cyber-physical systems (CPSs),
	e.g., \cite{SminSecure,2008Cyber,Bullo2018,tac2020wwsc,ren2020awide}.
	Nonetheless, it has been reported that such systems are often vulnerable to
	cyber-attacks \cite{2013Attack}, including false-data injection attacks
	\cite{FP-RC-FB:10p ,wu2019Switching}, replay attacks
	\cite{Zhu2017replay,replay2018}, and denial-of-service (DoS) attacks
	\cite{Cetinkaya2019overview}.  For instance, on February 8,
	2020, the telecommunication network of Iran suffered from DoS attacks for
	about an hour \cite{iran2020ddos}.  As a consequence, $25\%$ of the
	national internet connection dropped, leading to severe damage of critical
	infrastructure as well as significant economic loss.  In general, DoS
	attacks require little knowledge about the system and are therefore easy to
	be implemented.  Moreover, DoS attacks are destructive.  
	If an unstable open-loop
		process adopts a remote controller, then a long duration of DoS
	may render irreparable damages on physical systems as well as associated
	components.  These observations motivate the need for effective mechanisms
	to defend against DoS attacks and/or to mitigate the associated effect on
	physical entities.
	
	
	To this aim, a possible remedy is to design control strategies such that
	satisfied performance can be maintained regardless of the DoS attack
	strategies, which is referred to as {resilient control}.  This problem was
	first addressed by the work \cite{PersisInput}, in which a transmission
	policy along with some conditions on DoS attacks were developed for
	resilient stabilization of LTI systems.  Since then, a multitude of
	publications have studied resilient control of different systems,
	including, e.g., systems using an output-feedback controller in
	\cite{FengResilient,8880482,Liu2021resilient}, multi-agent systems in
	\cite{LuInput}, and nonlinear systems in \cite{Persis2016Networked}.
	
	It is worth emphasizing that all the aforementioned resilient control
	strategies build upon the model-based control approach developed in
	\cite{2003Model}.  In other words, they require an explicit system model,
	or need to perform a system identification step \emph{a
		priori}. Nonetheless, accurate system models may be challenging to
	acquire in real-world applications, while system identification of
	large-scale systems could be data hungry and/or computationally cumbersome.
	Data-driven methods, on the other hand, offer a new avenue for the control
	of unknown dynamic systems and pursue the design of a control law directly
	from data (i.e., from pre-collected input-output trajectories), {e.g., \cite{ren2021adata}.
		This method} also comes with rigorous theoretical guarantees thanks
	to the celebrated \emph{Fundamental Lemma} \cite{willems2005note}.  Indeed,
	there has been a recent line of research works exploring the fundamental
	lemma for data-driven control of unknown LTI systems.  Following the
	state-space description perspective, a parameterized model of linear
	feedback systems was developed such that state-feedback controllers can be
	directly obtained from solving data-dependent linear matrix inequalities in
	\cite{persis2020data,berberich2020robust,rueda2021data,zhang2019data,Talebi2020Online}.  Due to the popularity of model
	predictive controllers in industrial applications (e.g., see
	\cite{Mayne2000Constrained,Limon2008MPC,rawlings2019model}),
	another line of work dealt with data-driven MPC schemes (e.g., \cite{Coulson2019data,Coulson2021Bridging}).  
	Stability analysis and robustness guarantees of data-driven MPC against measurement
	noise have recently been studied in \cite{berberich2019data}.

	The goal of this present paper is to stabilize unknown LTI systems under
	DoS attacks based solely on input-output system trajectories acquired
	\emph{a priori} through some off-line experiments.  It has been shown in
	\cite{FengResilient} that LTI systems equipped with a predictor-based
	controller based on an explicit system model can achieve maximum resilience
	against DoS attacks. Specifically, an observer-based predictor is employed
	to maintain a reasonable estimate of the state even in the presence of DoS
	attacks; a feedback controller is then developed based upon this estimate.
	Nevertheless, both designs of the observer and of the feedback controller
	require an explicit system model, rendering such control methods difficult
	to apply when only input-output trajectories are available.  To overcome
	this challenge, a data-driven resilient MPC scheme is developed in this
	paper.  Resembling the data-driven MPC scheme presented in
	\cite{berberich2019data}, our resilient data-driven MPC generates a
	sequence of control inputs and associated predicted outputs by solving a
	convex optimization problem with purely data-dependent constraints, but
	only at time instants when new measurements are received (i.e., no DoS
	attack occurs).  On the other hand, at time instants when there is a DoS
	attack, the corresponding measurement transmission fails.  Our scheme takes
	the predicted control input from the most recent solution of the
	optimization problem, or simply a zero control input, depending upon
	whether the duration of DoS attacks exceeds the prediction horizon of the
	data-driven MPC.
	
	It is worth pointing out that there are a couple of notable differences
	between the data-driven MPC scheme in \cite{berberich2019data} and the one
	developed here.  To start with, in addition to DoS attacks, both process
	noise and measurement noise are modeled and accounted for in this paper,
	whereas only measurement noise was considered in \cite{berberich2019data}.
	Since both pre-collected measurements as well as the measurements received
	in the execution phase are corrupted by noise in \cite{berberich2019data},
	the resultant data-driven MPC involves solving a non-convex constrained
	optimization problem which can be computationally challenging (NP-hard in
	general).  In contrast, we consider only process noise in the
	data-collecting phase, while addressing both process and network-induced
	measurement noise in the execution phase.  Interestingly, this noise
	modeling approach leads to a convex data-driven MPC scheme in our paper,
	that is computationally appealing.  Although a convex data-driven MPC
	scheme was also developed in \cite{Coulson2019data}, no stability analysis
	and robustness guarantees were provided.  To the best of our knowledge,
	this is the first model-free (a.k.a. data-driven) controller that achieves
	the maximum resilience performance against DoS attacks as the model-based
	controller in \cite{FengResilient}.

	In a nutshell, the main contribution of the present work is threefold, summarized as follows.
	\begin{itemize}
		\item[\textbf{c1)}] Addressing both process and measurement noise, a
		data-driven MPC scheme is proposed for the control of LTI systems. Our
		scheme generates a series of control inputs and associated predicted
		outputs through solving a convex optimization problem driven by purely
		input-output data;
		\item[\textbf{c2)}] Leveraging the idea behind the predictor-based
		controller \cite{FengResilient}, the proposed data-driven MPC scheme is
		further modified to accommodate DoS attacks, thereby becoming a
		data-driven resilient controller;
		and,
		\item[\textbf{c3)}] Under standard conditions on DoS attacks and
		noise, local input-to-state stability (ISS) of the
		closed-loop system is established for the proposed data-driven
		resilient MPC, along with modifications to recover global ISS.
	\end{itemize}
	%
	\emph{Notation:}  Denote the set of real numbers by $\mathbb{R}$, and the set of integers by $\mathbb{Z}$. Given $\alpha \in \mathbb{R}$ or $\alpha \in \mathbb{Z}$, let $\mathbb{R}_{>\alpha} (\mathbb{R}_{\ge \alpha})$ or $\mathbb{Z}_{>\alpha} (\mathbb{Z}_{\ge \alpha})$ denote the set of real numbers or integers greater than (greater than or equal to) $\alpha$.
	Let $\mathbb{N}$ denote the set of natural numbers and define $\mathbb{N}_0 := \mathbb{N} \cup \{0\}$.
	For a matrix $M$, its left pseudo-inverse and right pseudo-inverse are denoted by $M^\dag$ and $M^\ddag$, respectively.
	Given a vector $x\in \mathbb{R}^{n_x}$, $\Vert x\Vert$ is its Euclidean norm, and for a positive definite matrix $P = P^T \succ 0$, define the weighted norm $\Vert x \Vert_P = \sqrt{x^\top P x}$.
	Let further $\Vert M\Vert$ be the spectral norm of matrix $M$.
	For discrete-time signals, define $\Vert w\Vert_{[t_1, t_2]} := \max_{t \in [t_1, t_2]}\{\Vert w_t\Vert\}$, and $\Vert w\Vert_{\infty} := \Vert w\Vert_{[0, \infty)}$
	Let $\underline{\lambda}_{P}$ ($\overline{\lambda}_{P}$) represent the minimum (maximum) eigenvalue of matrix $P$.
	Let $I$ represent the identity matrix with suitable dimensions, with $I_q$ denoting its $q$-th row, and $I_{[q,q+i]} = [I_q^\top,I_{q + 1}^\top,\cdots\!,I_{q +i}^\top]^\top$. 
	For $\delta \in \mathbb{R}_{>0}$, let $\mathbb{B}_{\delta} = \{x \in \mathbb{R}^{n_x}~|~\Vert x\Vert \le \delta \}$.
	Let $[t_1, t_2]$ denote the time interval from $t_1$ to $t_2$ with discrete time instants. 
	
	The Hankel matrix associated to the sequence $x_N := \{x_t\}_{t = 0}^{N - 1}$, is denoted by
	\begin{equation}
	H_{L}(x_N):=\left[
	\begin{matrix}
	x_{0} & x_{1} & \ldots & x_{N-L} \\
	x_{1} & x_{2} & \ldots & x_{N-L+1} \\
	\vdots & \vdots & \ddots & \vdots \\
	x_{L-1} & x_{L} & \ldots & x_{N-1}
	\end{matrix}
	\right].
	\end{equation}
	A stacked window of the sequence is given by
	\begin{equation}
	x_{[t_1, t_2]}=\left[
	\begin{matrix}
	x_{t_1} \\
	\vdots \\
	x_{t_2}
	\end{matrix}
	\right].
	\end{equation}
	For notational convenience, we will use $x$ to denote both the sequence itself and the stacked vector $x_{[0, N - 1]}$ containing all of its components, whenever clear from the context.
	\begin{figure}
		\centering
		\includegraphics[width=7cm]{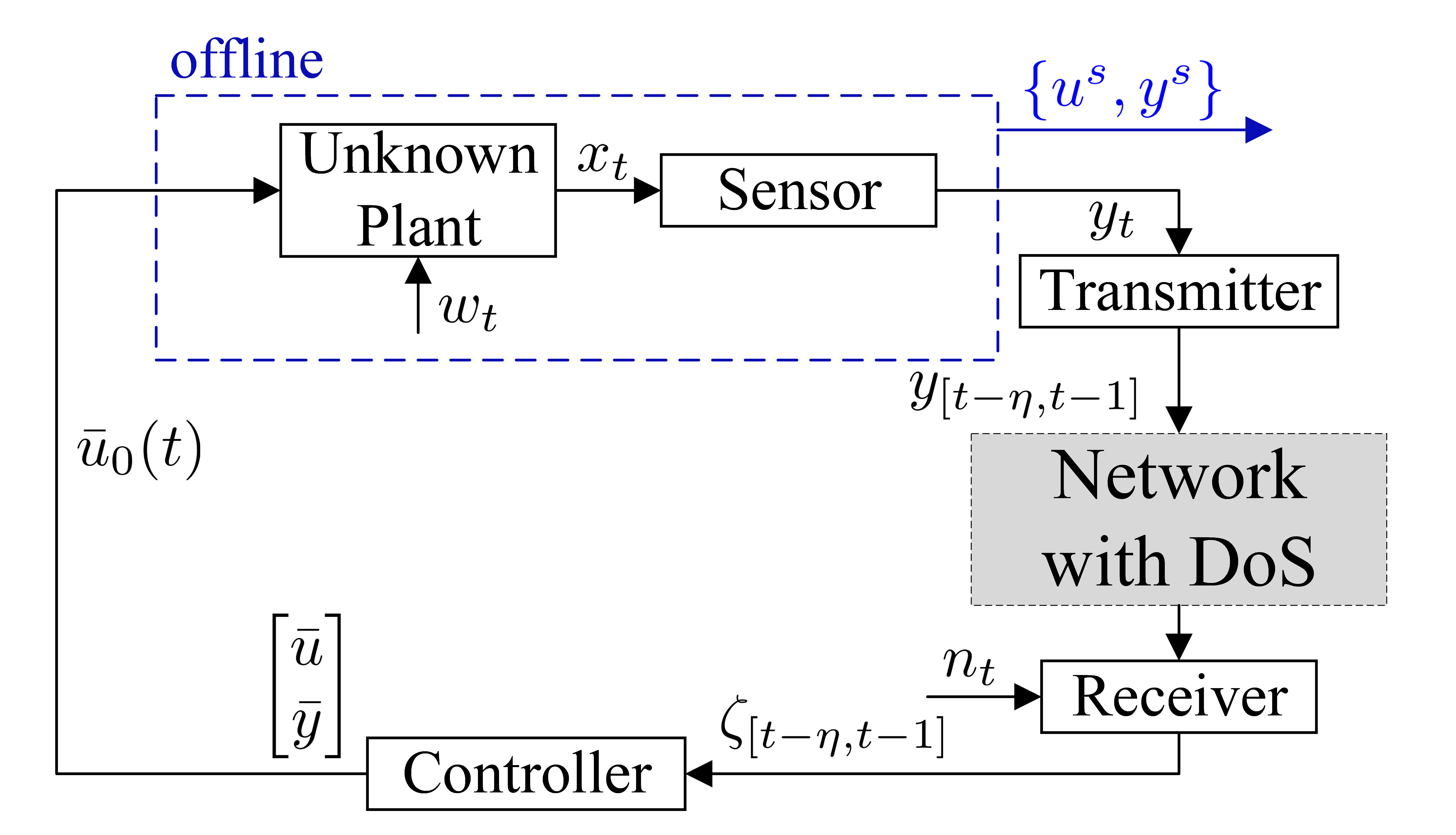}\\
		\caption{Closed-loop system using data-driven MPC.}\label{fig:system}
		\centering
	\end{figure}
	\section{Preliminaries and Problem Formulation}\label{sec:problemformulation}
	\subsection{Networked control systems}
	Consider the stochastic discrete-time linear time-invariant (LTI) system
	\begin{subequations}\label{eq:sys}
		\begin{align}
		x_{t+1} &= A x_{t} + B u_{t} + w_{t}\\
		y_{t} &= Cx_{t} + D u_{t}
		\end{align}
	\end{subequations}
	where $x_t \in \mathbb{R}^{n_x}$, $u_t \in \mathbb{R}^{n_u}$, $y_t \in \mathbb{R}^{n_y}$, and $w_t \in \mathbb{R}^{n_x}$ are the state, control input, output, and noise of the plant, respectively.
	The initial state $x_0$ is arbitrary.
	We make the following assumptions on system \eqref{eq:sys}.
	\begin{assumption}[\emph{Stabilizability and observability}]\label{as:ctrl}
		The pair $(A, B)$ is stabilizable, and the pair $(C, A)$ is observable.
	\end{assumption}
	\begin{remark}
		The observability index of the system \eqref{eq:sys} is denoted by $\eta$, i.e., $\eta := \min_{1 \le i\le n_x} \{{\rm rank}[C^\top, (CA)^\top, \cdots, $ $(CA^i)^\top]^\top = n_x\}$.
	\end{remark}
	\begin{assumption}[\emph{Unknown system model}]\label{as:sys}
		The system matrices $(A, B, C, D)$ in \eqref{eq:sys} are unknown, and only some input-output trajectories, i.e., $\{u^s_t, y^s_t\}_{t = 0}^{N - 1}$, obtained by some offline experiments are available.
	\end{assumption}
	
	At each time instant $t$, an output packet containing the past $\eta$ output measurements, i.e., $y_{[t - \eta, t - 1]}$, is sent to the remote controller through a communication channel subject to DoS attacks.
	As a result, not all packets can be received at the controller side, and details about this attack will be presented in the next section.
	Moreover, we consider that during transmission, the output is corrupted by an additive network-induced noise $n_t$, i.e.,
	\begin{equation}
	\zeta_t = y_t + n_t.
	\end{equation}
	\begin{assumption}[Noise bound]
			\label{as:noisebound}
			The process noise $w_t$ and the network-induced noise $n_t$ are bounded by a known constant $\bar{v} := \max_t\{\Vert w_t\Vert, \Vert n_t\Vert\}$ for all $t \in \mathbb{N}_0$.
	\end{assumption}
	On the other hand, the input $u_t$ generated by the controller goes through an ideal channel, i.e., the controller-to-plant channel is not affected by DoS attacks.
	See Fig. \ref{fig:system} for the depicted networked control system architecture.
	
	Due to the presence of noise $w_t$ and $n_t$, rather than the asymptotic stability, a weaker notion of stability, referred to as  input-to-state stability (ISS), is explored in this paper.
	The definition of ISS adapted from \cite[Definition 2.2]{Z1994Small} is presented as follows.
	\begin{definition}[\!\!{\emph{\cite[Definition 2.2]{Z1994Small}}}]
		\label{def:iss}
		The system \eqref{eq:sys} achieves input-to-state stability (ISS) if its solution satisfies
		\begin{equation}\label{eq:isps}
		\Vert x_t\Vert \le \beta\left(\Vert x_0\Vert, t\right) + \alpha\left(\Vert v\Vert_{\infty}\right),\quad \forall t \in \mathbb{N}_{\ge 0}
		\end{equation}
		where $\alpha$ is a function of class $\mathcal{K}_{\infty}$ \footnote{A function $\alpha : [0,\infty) \rightarrow [0, \infty)$ is said to be of class $\mathcal{K}$ if it is continuous, strictly increasing, and $\alpha(0) = 0$.
			A function $\alpha : [0,\infty) \rightarrow [0, \infty)$ is said to be of class $\mathcal{K}_{\infty}$ if it is of class $\mathcal{K}$ and also unbounded.}, and $\beta$ is a function of class $\mathcal{KL}$ \footnote{A function $\beta : [0,\infty)\times [0,\infty) \rightarrow [0, \infty)$ is said to be of class $\mathcal{KL}$ if $\beta(\cdot, t)$ is of class $\mathcal{K}$ for each fixed $t \ge 0$ and $\beta(s, t)$ decreases to $0$ as $t \rightarrow \infty$ for each fixed $s \ge 0$.}.
		If (\ref{eq:isps}) holds with $v = 0$, then the system is said to be globally asymptotically stable (GAS).
	\end{definition}
	
	Finally, we review the so-called uniform input-output-to-state stability
	(UIOSS) of a system introduced in \cite{CAI2008326}.
	\begin{proposition}[\emph{Uniform input-output-to-state stability}]
			\label{pro:uioss}
			If the system \eqref{eq:sys} is observable, there exist constants $c_1,
			c_2, c_3 > 0$, and $\epsilon \in (0, 1)$, such that for every $x \in
			\mathbb{R}^{n_x}$, $u \in \mathbb{R}^{n_u}$, and $w \in
			\mathbb{B}_{\bar{v}}$, the system is UIOSS with the entire
			trajectory obeying
			\begin{equation}\label{eq:uiossx}
			\Vert x_t\Vert \!\le\! c_1 \epsilon^{t + 1}\Vert x_0\Vert +c_2 \Vert u\Vert_{[0, j - 1]} + c_3 \Vert y\Vert_{[0, j - 1]}
			,~ \forall t\!\in\! \mathbb{N}_0
			\end{equation}
			for all $t\in\mathbb{N}_0$.
			In addition, there exists a matrix $P = P^T \succ 0$ such that for the UIOSS-Lyapunov function $V = x^T P x$, the following inequality holds for some constants $\sigma_1,\,\sigma_2,\,\sigma_3\!>\!0$
			\begin{equation}\label{eq:uiossV}
			V(x_{t + 1}) -V(x_t) \le -\sigma_1\Vert x _t\Vert^2 + \sigma_2\Vert u_t \Vert^2 + \sigma_2\Vert y_t \Vert^2
			\end{equation}
			for all $t\in\mathbb{N}_0$.
	\end{proposition}
	\subsection{Denial-of-Service attack}\label{sec:dos}
	Denial-of-Service  (DoS) attacks are attacks that are launched by malicious routers and jammers to block communication channels, which can lead to data packet losses, and are thus destructive.
	Evidently for an open-loop unstable plant, a long-duration of DoS can degrade the closed-loop system performance and even incur instability and divergence. 
	Therefore, defensive structure should be judiciously designed and equipped with the system to achieve stability in the presence of as well as resilience against DoS attacks.
	
	To rigorously evaluate the effectiveness of a defense method, a mathematical model characterizing DoS attacks should be introduced.
	Several models have been studied, and they can be roughly summarized as stochastic model and deterministic model. 
		For instance, Bernoulli processes \cite{yang2019mpc} and Markov processes \cite{Befekadu2015risk} are often used to model DoS attacks. 
		Another line along the stochastic model is using game-theoretic approaches to designing the attack strategy as well as the defense strategy simultaneously \cite{li2017sinr}.
		However, it is difficult to justify the incentive of an attacker using the stochastic model.
	To this aim, we call for a general deterministic model capitalizing on the DoS duration and the DoS frequency, that was initially presented in \cite{PersisInput}.
	For each $t \in \mathbb{N}_{0}$, upon introducing the following DoS indicator 
	\begin{equation}
	\ell_t := 
	\left\{
	\begin{array}{ll}
	0,& {\text{no DoS attack happens at}}~t\\
	1,&{\text{a DoS attack happens at}}~t
	\end{array}
	\right.,
	\end{equation}
	the DoS duration during the time interval $[t_1, t_2)$ is defined by $\Phi_d(t_1, t_2) = \sum_{i = t_1}^{t_2 - 1} \ell_i$.
	In addition, defining for each $t \in \mathbb{N}$
	\begin{equation}
	d_t := \left\{
	\begin{array}{ll}
	1, &\ell_{t} = 1~ {\text{and}}~\ell_{t-1} = 0\\
	0, &{\text{otherwise}}
	\end{array}
	\right.,
	\end{equation}
	then the DoS frequency during interval $[t_1, t_2)$ is expressed by $\Phi_f(t_1, t_2) = \sum_{i = t_1}^{t_2 - 1} d_i$.
	
	It is self-evident from the definitions that, if both the DoS duration and the DoS frequency are considerably sizable to prevent all packets from being transmitted, no meaningful control input can be constructed to stabilize the plant.
	As a consequence, some assumptions on the DoS frequency and DoS duration should be made for practical investigation purpose of system stability in the presence of DoS attacks.
	\begin{assumption}[\emph{DoS frequency}]\label{as:dosfre}
		There exist constants $\kappa_f \in \mathbb{R}_{\ge 0}$, and $\nu_f \in
		\mathbb{R}_{\ge 2}$, also known as chatter bound and average dwell-time,
		respectively, such that the DoS frequency satisfies
		\begin{equation}\label{eq:dosfre}
		\Phi_f(t_1, t_2) \le \kappa_f + \frac{t_2 - t_1}{\nu_f}
		\end{equation}
		over every time interval $[t_1, t_2)$, where $t_1\le t_2 \in \mathbb{N}_{0}$.
	\end{assumption}
	
	\begin{assumption}[\emph{DoS duration}]\label{as:dosdur}
		There exist constants $\kappa_d \in \mathbb{R}_{\ge 0}$, and $\nu_d \in
		\mathbb{R}_{\ge 1}$, also known as chatter bound and average duration
		ratio, respectively, such that the DoS duration satisfies
		\begin{equation}\label{eq:dosdur}
		\Phi_d(t_1, t_2) \le \kappa_d + \frac{t_2 - t_1}{\nu_d}
		\end{equation}
		over every time interval $[t_1, t_2)$, where $t_1\le t_2 \in \mathbb{N}_{0}$.
	\end{assumption}
	
	In fact, the original version of modeling DoS attacks in terms of DoS
	frequency and duration was established in \cite{PersisInput} to address the
	stabilization problem of a continuous-time system under DoS attacks. These
	notions have been popularly used in the literature; see e.g.,
	\cite{FengResilient,LuInput,FengNetworked,PersisInput,Persis2016Networked}.
	When studying discrete-time systems, a discrete-time DoS model that is
	similar to the present one, was discussed in \cite{8880482}.
	In their work, LTI systems without noise were considered, and thus weaker
	assumptions were used.
	To be specific, they constrained the DoS frequency and duration only on the interval $[0, t),
	t\in \mathbb{N}_0$, rather than on every sub-interval $[t_1, t_2)$
	of $[0, t)$.
	
	%
	\begin{remark}[\emph{Implications of DoS parameters}]
		Taking a switching systems perspective, the quantity $\nu_f$ in Assumption \ref{as:dosfre} can be interpreted as the average dwell-time \cite{Hespanha1999STABILITY} between two consecutive DoS attacks \emph{off/on} switches over the time interval $[t_1, t_2)$.
		On the other hand, Assumption \ref{as:dosdur} requires that the average duration of DoS attacks  not exceed a fraction $1/\nu_d$ of the entire interval.
		The constants $\kappa_f$ and $\kappa_d$ are additional regularization parameters that can be chosen to attain tighter bounds on the DoS duration and on the DoS frequency, respectively.
		To see this, suppose a DoS attack happens at time $t_1$. Then, it is obvious for the time interval $[t_1, t_1 + 1)$, that $\Phi_f(t_1, t_2) = 1$, and $t_2 - t_1 = 1$.
		By virtue of the fact that $\nu_f \in \mathbb{R}_{\ge 2}$, we deduce that $(t_1 + 1 - t_1)/\nu_f \le 1/2 < 1$, which suggests that a constant $\kappa_f \ge 1$ is required to ensure that \eqref{eq:dosfre} holds true for every time interval.
		The same reasoning can be conducted for $\kappa_d$.
	\end{remark}
	As a direct implication of Assumptions \ref{as:dosfre} and \ref{as:dosdur}, it can be deduced that the number of time steps between two successful transmissions is upper bounded.
	Let $\{s_r\}_{r\in \mathbb{N}_0}$ collect the successful transmission time instants, at which $\ell_{s_r} = 0$.
	\begin{lemma}[\!\!
		{\cite[Lemma 3]{FengResilient}}]
		\label{lem:dos}
		Suppose that the DoS attacks satisfy Assumptions \ref{as:dosfre} and \ref{as:dosdur} with
		\begin{equation}\label{eq:doscondition}
		\frac{1}{\nu_f} + \frac{1}{\nu_d} < 1.
		\end{equation}
		Then it holds that $s_0 \le T - 1$, and $s_{r + 1} - s_{r} \le T$ for all $r\in \mathbb{N}_0$ with
		\begin{equation}\label{eq:Ts}
		T := (\kappa_d + \kappa_f)\Big(1 - \frac{1}{\nu_d} - \frac{1}{\nu_f}\Big)^{-1} + 1.
		\end{equation} 
	\end{lemma}
	\begin{remark}[{\emph{Maximum resilience}}]
		Condition \eqref{eq:doscondition} is referred to as the
		\emph{maximum resilience} against DoS attacks one can achieve for an
		open-loop unstable system \cite[Lemma 3]{FengResilient}.  In other words,
		if \eqref{eq:doscondition} does not hold, i.e., ${1}/{\nu_f} +
		{1}/{\nu_d} \ge 1$, then one can always design a sequence of DoS attacks
		to render the system unstable whatever control strategy is used.  For
		example, consider a sequence of DoS attacks under which $l_t = 1$ for all
		$t \in \mathbb{N}_0$.  It is easy to check that this sequence of attacks
		satisfies Assumptions \ref{as:dosfre} and \ref{as:dosdur} with $\nu_f =
		1$, $\nu_d = \infty$, $\kappa_d = 1$, and $\kappa_f = 1$.  In addition,
		it can be obtained that ${1}/{\nu_f} + {1}/{\nu_d} = 1$.  However, it is
		impossible to stabilize an open-loop unstable system under such attacks
		as there will be no successful transmissions from the system to construct
		meaningful control input signals.
	\end{remark}
	\subsection{Fundamental Lemma}\label{sec:fundamental}
	The objective of this present paper is to design resilient controllers for stabilization of an unknown LTI system under DoS attacks and additive noise using only measured input-output data.
	Commonly, if the system matrices $(A, B, C, D)$ were precisely known, a number of proposals have been presented for resilient control of LTI systems \eqref{eq:sys} under DoS attacks obeying Assumptions \ref{as:dosfre} and \ref{as:dosdur}. A natural approach is to endow the closed-loop control system with prediction capabilities such that the missing measurements can be reconstructed (predicted) during DoS \cite{FengResilient}.
	To the best of our knowledge, no previous works have dealt with the stabilization problem of \emph{unknown} systems in the presence of DoS attacks. 
	There are three emerging challenges: i) how to characterize an LTI system as well as infer the wanted quantity from observed noisy input-output data; 
	ii) how to design a resilient controller against DoS attacks based only on data; and, iii) the associated stability analysis and robustness guarantees.
	
	To address the first challenge, we invoke the well-known Fundamental Lemma, that was initially discovered in \cite{willems2005note} and subsequently generalized in \cite{van2020extension} and \cite{fundamental2021}. 
	Before formally presenting the Fundamental Lemma, 
	the standard definition of persistency of excitation  is introduced first.

	\begin{definition}[\emph{Persistency of excitation}]\label{def:pe}
		A sequence $u_N := \{u_t\in\mathbb{R}^{n_u} \}_{t=0}^{N-1}$ is said to be persistently exciting of order $L$ if ${\rm {rank}}(H_L(u_N)) = n_u L$.
	\end{definition}

	Based on Definition \ref{def:pe}, it has been shown in \cite{willems2005note} that any input-output trajectory can be expressed as a linear combination of pre-collected persistently exciting input-output data, which is also known as the Fundamental Lemma.
	\begin{lemma}[\emph{Fundamental Lemma \cite{willems2005note}}]\label{lem:fundamental}
		Consider the noise-free version of the LTI system \eqref{eq:sys} described by
		\begin{subequations}\label{eq:idealsys}
			\begin{align}
			x_{t+1} &= A x_{t} + B u_{t}\\
			y_{t} &= Cx_{t} + D u_{t}.
			\end{align}
		\end{subequations}
		Suppose that the input-output sequence $\{\bar{u}_N^s, \bar{y}_N^s\}:=\{\bar{u}_t^s, \bar{y}_t^s\}_{t = 0}^{N - 1}$ is a trajectory of the system \eqref{eq:idealsys}
		, induced by a persistently exciting input sequence $\bar{u}^s$ of order $L + n_x$.
		Then, $\{\bar{u}, \bar{y}\}:=\{\bar{u}_t, \bar{y}_t\}_{t = 0}^{L - 1}$ is a trajectory of the system \eqref{eq:idealsys} if and only if there exists a vector $g \in \mathbb{R}^{N - L + 1}$ such that the following holds
		\begin{equation}\label{eq:fundamental}
		\left[
		\begin{matrix}
		H_{L}\!\left(\bar{u}^{s}_N\right) \\
		H_{L}\!\left(\bar{y}^{s}_N\right)
		\end{matrix}
		\right] g=\left[
		\begin{matrix}
		\bar{u} \\
		\bar{y}
		\end{matrix}
		\right].
		\end{equation}
	\end{lemma}
	Regarding the Fundamental Lemma, two remarks come ready.
	\begin{remark}[{\emph{Requirements of Fundamental Lemma}}]
		Lemma \ref{lem:fundamental} indicates that the length $L$ of the constructed trajectory $\{\bar{u}, \bar{y}\}$ depends on the persistently exciting order of the pre-collected data $\{\bar{u}_t^s, \bar{y}_t^s\}_{t = 0}^{N - 1}$.
		More importantly, it offers a feasible way to design a data-dependent model for the noise-free LTI system \eqref{eq:idealsys}. When dealing with the noisy system \eqref{eq:sys}, some modifications should be made.
	\end{remark}
	\begin{remark}[{\emph{Different ways of pre-collecting trajectories}}]
		The sequence $\{\bar{u}_N^s, \bar{y}_N^s\}$ can be obtained from either a single long enough trajectory or formed by multiple short trajectories obtained by simulating the system using different initial conditions and input sequences.
		In fact, both ways of collecting data are equivalent; please refer to \cite{van2020extension} for details.
	\end{remark}
	\subsection{Model-based observer-based controller}\label{sec:resilient}
	When the system matrices of \eqref{eq:sys} are perfectly known,  model-based control strategies for achieving stability under DoS attacks were designed in, e.g., \cite{FengResilient,sun2020resilient,yang2019mpc}. 
		In \cite{sun2020resilient,yang2019mpc}, although MPC scheme were used for resilient control, neither the system model nor the attack model resemble that in the present work.
		On the other hand, the setting in \cite{FengResilient} is the same as this paper.
		Moreover, system adopting the model-based observer-based control strategy in \cite{FengResilient} achieves maximum resilient against DoS attacks.
	The key idea behind this strategy is to equip an observer at the sensor side to estimate the state $x_t$ at every instant, and adopt a predictor-based state-feedback controller using the predicted state $\hat{x}_t$-based control law $u_t = K\hat{x}_t$.
		When there is no DoS attack (i.e., $t=s_r$), the predictor receives new estimated state $\bar{x}_t$ successfully from the observer, and updates the predicted state $\hat{x}_t$; during DoS attacks (i.e., $t \ne s_r$) however, the predictor simply updates the predicted state following a prediction step based on the system model. Mathematically, the observer is given by 
		\begin{subequations}\label{eq:observer}
			\begin{align}
			\bar{x}_{t+1} &= A \bar{x}_{t} +L(y_t - C\bar{x}_{t}) + B u_t\label{eq:observer1}\\
			u_t &= K \hat{x}_t \label{eq:observer2}
			\end{align}
		\end{subequations}
		and the predictor-based controller is given by
		\begin{subequations}\label{eq:predictor}
			\begin{align}
			\hat{x}_{t+1}& = A \hat{x}_t + B u_t, \label{eq:predictor1}\\
			\hat{x}_{t}& = \bar{x}_{t},~~& t = s_r\label{eq:predictor2}\\
			u_t &= K \hat{x}_t \label{eq:predictor3}
			\end{align}
		\end{subequations}
		where matrix $L$ is a deadbeat observer gain matrix such that $(A - LC)^{\eta} = 0$, and $\eta$ is the observability index of the system \cite{OreillyJ}; and
		$K$ is any stabilizing feedback gain matrix, i.e.,  $A + BK$ is Schur stable.
	
	It has been shown in \cite{FengResilient} that, for DoS attacks satisfying \eqref{eq:doscondition}, the error between $\hat{x}_t$ and ${x}_t$ is reset to a constant depending on the noise bound and the system matrices at least once every $T$ time instants. 
	Indeed, this resetting is guaranteed because there exists at least one  successful transmission, in any sequence of $T$ time instants, as guaranteed by Lemma \ref{lem:dos}. 
	
	Without loss of generality, consider the interval $[0, T-1]$, with $s_0 = T - 1$ being the first successful transmission.
	Define the error $\bar{e}_{t} := \bar{x}_t - x_t$, and the error $\hat{e}_{t} := \hat{x}_t - x_t$.
	It can be obtained from \eqref{eq:sys}, \eqref{eq:observer}, and \eqref{eq:predictor} that
	\begin{subequations}\label{eq:error}
		\begin{align}
		\bar{e}_{t+1} &= (A - LC) \bar{e}_{t} +L n_t -  w_t\label{eq:error1}\\
		\hat{e}_{t+1} &= A \hat{e}_{t} - w_t,~~~ &t \ne s_r\label{eq:error2}\\
		\hat{e}_{t} &= \bar{e}_{t},~~~ &t = s_r\label{eq:error3}.
		\end{align}
	\end{subequations}
	For $t = s_0$, it follows from the recursion in \eqref{eq:error} that 
	\begin{subequations}
		\begin{align}
		\hat{e}_{s_0} = \bar{e}_{s_0}& = (A - LC)^{\eta}\bar{e}_{s_0 - \eta }\nonumber\\
		&~~~~ - \sum_{i = 0}^{\eta - 1} (A - LC)^{i}(Ln_{s_0 - i-1} - w_{s_0- i-1})\label{eq:errorprop1}\\
		&=  - \sum_{i = 0}^{\eta - 1} (A - LC)^{i}(Ln_{s_0 - i-1} - w_{s_0- i -1}) \label{eq:errorprop2}
		\end{align}
	\end{subequations}
	where \eqref{eq:errorprop2} follows because $(A - LC)^{\eta} = 0$.
	Moreover, since matrices $A$, $
	L$, $C$, and noise $w_t$, $n_t$ are all finite and bounded, then the summation of products of finite moments of these terms is also bounded. 
	This implies that when there is a successful transmission, the error is reset to a bounded value; that is, the estimated state matches the actual state well. 
	For DoS attacks obeying \eqref{eq:doscondition}, Lemma \ref{lem:dos} guarantees that  this \emph{``predict-then-reset''} cycle of the estimation error, happens frequently (in fact, at least once every $T$ time instants).
	In conclusion, under our working assumptions on the DoS attacks, the estimation error of the predictor in \eqref{eq:predictor} with the selected parameters is always bounded, and the system can be stabilized.
	Obviously, successful application and implementation of the predictor-based controller in \eqref{eq:predictor} to defend against DoS attacks hinges on perfect knowledge of the system model. When this knowledge is not available, 
	although standard system identification approaches can be used to obtain an estimate of the model, 
	this step can be time-consuming and computationally expensive as the size of the system increases. For this reason, data-driven methods, deriving control strategies directly from data without performing explicit system identification procedures, have become prevalent recently.
	\section{Data-driven Resilient Control under DoS}
	In this section, we address the challenging stabilization problem of
	unknown LTI systems in the presence of DoS attacks. Our proposal is to
	redesign the robust predictor-based controller in \eqref{eq:predictor}
	leveraging the Fundamental Lemma \eqref{eq:fundamental}, to develop a novel
	resilient control scheme that capitalizes purely on data.
	Specifically, a data-dependent controller is first constructed to follow
	the predict-then-reset circle, which resembles the data-driven MPC developed in \cite{berberich2019data}.
	Capitalizing on this controller, a DoS-resilient MPC algorithm is proposed
	and subsequently corroborated by a stability analysis.

	\subsection{Data-driven resilient predictive control}
	According to Section \ref{sec:resilient}, a known LTI system under DoS attacks can be stabilized by the predictor-based controller \eqref{eq:predictor}. 
	Without knowing the system matrices $A$ and $B$, 
	MPCs are preferred instead of state-feedback controllers in which closed-loop stabilizing feedback gain matrices  $K$ are in general hard to find. 
	Traditional MPC schemes (e.g., \cite{rawlings2019model, Mayne2000Constrained}) can predict future trajectories based on the system model, initial conditions, and using some terminal constraints.
		The optimal future trajectory is generated by minimizing a cost function.
		Choosing a control input from the optimal future trajectory, the system can be stabilized.
	Thanks to Lemma \ref{lem:fundamental}, any trajectory of an LTI system can be exactly characterized by a (single) input-output trajectory that is persistently exciting of enough order.
	This naturally inspires one to construct a data-driven MPC scheme, that is predicting future trajectories using some input-output collections to replace a system model.
	To sum up, our idea is to replace the model-based controller in \eqref{eq:predictor} with a data-driven MPC, which builds solely on data consisting of a pre-collected input-output trajectory, initial conditions, as well as terminal constraints.

	Before formally introducing our data-driven MPC, a prerequisite assumption on the pre-collected data is posed.
	Let $L$ denote the prediction horizon of the MPC scheme.
	At the time instant when no DoS attack occurs (i.e., $t = s_r$), the MPC scheme predicts a trajectory of  $L$ steps in the future, i.e., from $t$ to $t + L - 1$, and only the input obtained for time $t$ will be used to control the system.
	According to Lemma \ref{lem:fundamental}, to be able to generate a predicted trajectory of length $L$, the pre-collected data should be persistently exciting of at least order $L$.
	To validate this requirement, we make the following assumption on the pre-collected data.
	\begin{assumption}[{\emph{Pre-collected data}}]\label{as:multitre}
		Let sequence $\{u^{s}_{N},\! y^{s}_{N}\} := \{u^{s}_{t}, y^{s}_{t}\}_{i = 0}^{N - 1}$ denote an input-output trajectory generated by the LTI system \eqref{eq:sys} from initial condition $x_0^{s}$, in which the input sequence $\{u^{s}_N\}$ is persistently exciting of order $L + n_x + \eta$, and the output sequence $\{y^{s}_N\}$ is collected offline without network-induced noise $n_t$.
	\end{assumption}

	Based on the assumption above, we build on the work of \cite{berberich2019data} and advocate the following data-driven MPC for control of unknown LTI systems without DoS attacks, which is solved at every time $t\in\mathbb{N}_0$.
	\begin{subequations}\label{eq:mpc}
		\begin{align}
		J^*_L(u_{[t - \eta, t - 1]}, & \zeta_{[t - \eta, t - 1]}) := \nonumber \\
		\underset{g(t), h(t)\atop
			\bar{u}_i(t), \bar{y}_i(t)}{\min}
		~&\sum_{i = 0}^{L - 1} \ell(\bar{u}_i(t), \bar{y}_i(t)) \!+\! \lambda_{g}\bar{v} \Vert g(t)\Vert^2 \!+\! \frac{\lambda_h}{\bar{v}}\Vert h(t)\Vert^2 \nonumber\\
		{\rm s.t.}\quad \;&
		\left[
		\begin{matrix}
		\bar{u}(t)\\
		\bar{y}(t) + h(t)
		\end{matrix}
		\right] = 
		\left[
		\begin{matrix}
		{H}_{L + \eta}(u^s)\\
		{H}_{L + \eta}(y^s)
		\end{matrix}
		\right] g(t), \label{eq:mpc1}\\
		&	\left[
		\begin{matrix}
		\bar{u}_{[-\eta, -1]}(t)\\
		\bar{y}_{[-\eta, -1]}(t)
		\end{matrix}
		\right] = 
		\left[
		\begin{matrix}
		u_{[t - \eta, t - 1]}\\
		\zeta_{[t - \eta, t - 1]}
		\end{matrix}
		\right], \label{eq:mpc2}\\
		&	\left[
		\begin{matrix}
		\bar{u}_{[L -\eta, L-1]}(t)\\
		\bar{y}_{[L-\eta, L-1]}(t)
		\end{matrix}
		\right] = 
		\left[
		\begin{matrix}
		0\\
		0
		\end{matrix}
		\right],\label{eq:mpc3}\\
		&~ \bar{u}_i \in \mathbb{U},~~ i \in [0,L - 1]\label{eq:mpc4}.
		\end{align}
	\end{subequations}
	where
	\begin{itemize}
		\item [1)]
		Constraint \eqref{eq:mpc1} is reminiscent of \eqref{eq:fundamental} in Lemma \ref{lem:fundamental}, where $\bar{u}(t) = [\bar{u}_{-\eta}^\top(t), \cdots\!, \bar{u}^\top_{L}(t)]^\top\in \mathbb{R}^{n_u(L + \eta)}$ and $\bar{y}(t) = [\bar{y}_{-\eta}^\top(t), \cdots\!, \bar{y}^\top_{L}(t)]^\top\in \mathbb{R}^{n_y(L + \eta)}$. To sustain feasibility of the optimization problem,
		a slack vector $h(t) = [h_{-\eta}^\top(t), \cdots\!, h^\top_{L}(t)]^\top\!\in\! \mathbb{R}^{n_y(L + \eta)}$ is introduced to compensate for the network-induced noise $n_t$;
		\item [2)]
			Constraints \eqref{eq:mpc2} and \eqref{eq:mpc3} use $\eta$ input and output pairs to restrict the state $x$ at time instants $t$ and $t + L - 1$, which can be implemented thanks to the observability of the system.
			To be specific, Constraint \eqref{eq:mpc2} is imposed to maintain the continuity of the true trajectory at time $t$, and Constraint \eqref{eq:mpc3} ensures convergence of the predicted trajectories;	
		\item [3)]
		Constraint \eqref{eq:mpc4} indicates that for $i \in [0, L - 1]$, the control inputs $\bar{u}_i$ are chosen from a given convex set $\mathbb{U}$ with $0 \in \mathbb{U}$, e.g., $\mathbb{U} = [-u_{\max}, u_{\max}]$ for some $u_{\max}$;
		\item[4)]
		The cost function $\ell (\bar{u}, \bar{y})$ is a quadratic stage cost given by $\ell (\bar{u}, \bar{y}) = \Vert \bar{u} \Vert^2_{R_1} + \Vert \bar{y} \Vert^2_{R_2}$, where $R_1\succ 0 $, and $ R_2 \succ 0$; the penalty term $\|h(t)\|^2$ enforces sparing use of slacks to produce solutions of
			minimal constraint violation; and the coefficient $\lambda_h > 0$ balances between minimizing the cost and penalizing the constraint violation; and $\|g(t)\|^2$ with coefficient $\lambda_g > 0$ is introduced to restrain the effect of process noise $w_t$.
			Notice that the regularization coefficients of $\|h(t)\|^2$ and $\|g(t)\|^2$ also depend on the upper bound on the noise. 
			This indicates that
			i) the slack variable $h(t)$ decreases with the noise level $\bar{v}$; and ii) $\|h(t)\|^2$ is small enough compared with the $\lambda_g \bar{v}^2\|g(t)\|^2/\lambda_h$, which is required to prove the system stability.
	\end{itemize}
	
	It is easy to see that all the constraints and the objective function are convex, so
	Problem \eqref{eq:mpc} is convex and can be solved efficiently by means of off-the-shelf convex programming solvers. 
	
	\begin{remark}[{\emph{Comparison with the work \cite{berberich2019data}}}]
		Difference between our data-driven MPC scheme in \eqref{eq:mpc} and that of \cite{berberich2019data} lies in two aspects.
		First, instead of an ideal system, we consider LTI systems containing process noise $w_t$. 
		In this setting, the Hankel matrix is a noisy one that does not contain exact system trajectories.
		Therefore, to invoke the Fundamental Lemma, inspired by \cite{krishnan2021On}, the error between the Hankel matrix constructed by the ideal input-output trajectory and the actual trajectory is characterized in the following proof.
		Moreover, the data-driven MPC problem in \cite{berberich2019data} is non-convex, which is {NP-hard} in general and computationally challenging. 
		Inspired by \cite{berberich2021robust}, the cost function $J^*_L$ in the present work is modified by adding noise bound as the coefficients of $\|h(t)\|^2$ and $\|g(t)\|^2$.
			In this manner, 
			the non-convex constraint $\Vert h_i(t)\Vert_{\infty} \le \bar{n} (1 + \Vert g(t)\Vert_1)$ 
			(where $\Vert \cdot \Vert_1$, and $\Vert \cdot \Vert_{\infty}$ are the $\ell_{1}$-, and $\ell_{\infty}$-norm of vectors, respectively)
			employed in \cite{berberich2019data} to guarantee closed-loop stability for their data-driven MPC, is not needed.
			Furthermore, the data-driven MPC scheme in \eqref{eq:mpc} can also guarantee the system stability if the  pre-collected output data is corrupted by noise, i.e., $\tilde{y}^s = y^s + \tilde{n}$.
			In this setting, a similar stability analysis can be conducted yet involving more complicated noise terms.
			For ease of exposition, the present work assumes that the pre-collected output is free from noise.
		
	\end{remark}
	Upon solving 
	\eqref{eq:mpc}, a future trajectory of length $L$ is obtained.
	In addition, due to constraints \eqref{eq:mpc2} and \eqref{eq:mpc3}, length $L$ should be long enough.
	Therefore, the following assumption is made.
	\begin{assumption}[{\emph{Prediction horizon}}]\label{as:horizon}
		The problem horizon satisfies $L \ge \eta + n_x$.
	\end{assumption}

	It is worth stressing that the proposed MPC in \eqref{eq:mpc} requires no system matrices but only some input-output trajectories collected from the system \eqref{eq:sys} by means of off-line experiments.
	However, when the linear system is subject to DoS attacks, one may not be able to solve \eqref{eq:mpc} at every time instant, since packets containing the last $\eta$ outputs may not successfully arrive at the controller.
	To address the challenge emerged with missing packets due to the DoS attacks, the proposed data-driven MPC scheme is modified by virtue of the idea behind the predictor-based resilient controller in Section \ref{sec:resilient}.
	Specifically, when there is
	no DoS attack (i.e., $t = s_r$), the controller receives new output packets successfully from the plant,  solves Problem \eqref{eq:mpc}, and uses the  input obtained for time $t$ as the control input. However,
	when a DoS attack occurs (i.e., $t \ne s_r$), if $t - s_r \le L$, 
	it simply sends the input for time $t$ from the most recent solution of Problem \eqref{eq:mpc}; that is, for $t \in [s_r, s_r + L]$, the first $t - s_r$ computed inputs $\bar{u}_i(s_r) \in [0, t - s_r]$ are to be used sequentially when there is a DoS attack.
	On the other hand, if $t >s_r + L$, i.e., no  control input for time $t$ has been computed from previous solving of \eqref{eq:mpc}, then zero inputs will be used until the next successful transmission comes and Problem \eqref{eq:mpc} is solved again for new control inputs.
	
	For reference, our proposed resilient data-driven control scheme is summarized in Algorithm \ref{alg:mpc}.
	
	\begin{algorithm}[h]
		\caption{Data-driven resilient control under DoS attacks.}
		\label{alg:mpc}
		\begin{algorithmic}[1]
			\STATE {\bfseries Input:} Predict horizon $L \ge  \eta + {n_x}$; parameters of the cost function $R_1 \succ 0$, $R_2 \succ 0$, $\lambda_g > 0$ and $\lambda_h > 0$;
			noise bound $\bar{v}$;
			input-output trajectories $\{(u^{s}_{0, N - 1}, y^{s}_{0, N - 1})\}$ of system \eqref{eq:sys} from initial condition $x_0^{s}$, where input sequence $\{u^{s}\}$ is persistently exciting of order $L + 2 \eta$.
			\STATE {\bfseries Construct} Hankel matrix for the input-output trajectory, i.e., $H\! =\! [H^\top_{L + \eta}(u^{s}), H^\top_{L + \eta}(y^{s})]^\top$.
			\STATE {\bfseries If} $t = s_r$, do \label{alg:mpc1}				
			\STATE \quad Use the past $\eta$ measurements, i.e.,  $u_{[t - \eta, t - 1]}$ and $\zeta_{[t - \eta, t - 1]}$, to solve Problem \eqref{eq:mpc}.
			Set $u_t = \bar{u}_0(t)$.
			\STATE {\bfseries Else if} $t \ne s_r$
			\STATE {\bfseries \quad if} $t - s_r \le L$
			\STATE \quad \quad Set $u_t = \bar{u}_{t - s_r}(s_r)$.\label{alg:mpcu1}
			\STATE {\bfseries \quad else if} $t - s_r > L$ \label{alg:mpcL}
			\STATE \quad \quad Set $u_t = 0$.\label{alg:mpcu2}
			\STATE {\bfseries \quad End if}
			\STATE {\bfseries End if}
			\STATE {\bfseries Set} $t = t + 1$ and go back to \ref{alg:mpc1}.
		\end{algorithmic}
	\end{algorithm}

	\subsection{Stability analysis}\label{sec:stable}
	In this following, performance of the proposed data-driven resilient controller is analyzed.  
	The stability analysis follows the same line as that of model-based resilient controller, in e.g., \cite{FengResilient}.
	It proceeds in two steps. Similar to the model-based controller in Section \ref{sec:resilient}, we first show that the error between the predicted output by solving Problem \eqref{eq:mpc} and the actual output is always bounded.
	Next, we construct a Lyapunov function, and establish the stability of system \eqref{eq:sys} with the devised data-driven resilient controller under conditions on the DoS attacks and level of noise.
	
	Before proceeding, we define some variables that will be frequently used throughout the proof.
	According to \eqref{eq:mpc2}, at each time  instant $t$, the most recent $\eta$ inputs and outputs, i.e., $[u^\top_{[t - \eta, t - 1]}, \zeta^\top_{[t - \eta, t - 1]}]^\top$, are used to initialize (the trajectory of) Problem \eqref{eq:mpc}, so its solution depends on $[u^\top_{[t - \eta, t - 1]}, \zeta^\top_{[t - \eta, t - 1]}]^\top$.
	Moreover, it has been proved in \cite{1991Linear} that the state $x_{t - \eta}$ of system \eqref{eq:sys} can be derived from $u_{[t - \eta, t - 1]}$ and $y_{[t - \eta, t - 1]}$ when system matrices $(A, B, C)$ are known, and  pair $(C,A)$ is observable.
	are known.
	To this end, let us define the following augmented state vectors of the system 
	\begin{equation}
	z_t := \left[
	\begin{matrix}
	u_{[t - \eta, t - 1]}\\
	y_{[t - \eta, t - 1]}
	\end{matrix}
	\right],\qquad 
	\tilde{z}_t := \left[
	\begin{matrix}
	u_{[t - \eta, t - 1]}\\
	\zeta_{[t - \eta, t - 1]}
	\end{matrix}
	\right].
	\end{equation}
	In addition, let $\Gamma_z$ denote the linear transformation from $z$ to an arbitrary but fixed state $x$ that obeys a similar transformation expressed analytically using $(A, B, C)$, i.e., $x_t = \Gamma_z z_t$, and we have that $\Vert x_t\Vert \le \gamma_z \Vert z_t\Vert$, where $\gamma_z := \Vert\Gamma_z\Vert$. 
	
	The following lemma provides an upper bound on the error between the predicted output $\bar{y}$ and the actual output $y$.
	\begin{lemma}[{\emph{Bounded output prediction error}}]\label{lem:bounde}
		Consider the system \eqref{eq:sys} with the controller in Algorithm \ref{alg:mpc}.
		Suppose that Assumptions \ref{as:ctrl}--\ref{as:horizon} hold.
		If i) DoS attacks satisfy Assumptions \ref{as:dosfre} and \ref{as:dosdur} with parameters $\nu_d$ and $\nu_f$ obey \eqref{eq:doscondition} and, ii) Problem \eqref{eq:mpc} is feasible at $s_r$, and $J^*_{L}(\tilde{z}_{s_r}) \le \bar{J}$, then the error  $e_{y,t} := y_t - \bar{y}^*_{t - s_r}(s_r)$ between the predicted output contained in the solution of \eqref{eq:mpc} and the actual output satisfies
		\begin{equation}\label{eq:boundey}
		\Vert e_{y,s_r + q}\Vert^2 \le \beta_1(\bar{v}, q), \quad 
		q \in [0,s_{r + 1} - s_{r})
		\end{equation}
		with a $\mathcal{KL}$-function $\beta_1(\bar{v}, q)$ defined by 
			\begin{align}
			&\beta_1(\bar{v}, q) :=b_3 \xi^{q + \eta}\sqrt{\eta}\bar{v}+ b_2 \xi^{q + \eta}\Bigg[\sqrt{\eta}\bar{v} + \sqrt{\bar{J}\bar{v}/\lambda_h} \label{eq:beta_1}\\
			&+ \Vert \Upsilon_{\eta}(I) \Vert \sqrt{\frac{\eta(N - L - \eta + 1)\bar{J}\bar{v}}{\lambda_g}}\Bigg]  +b_1\!\!\!\!\! \sum^{q + \eta - 1}_{j = 0}\!\!\!\!\xi^{j} \bar{v} + \sqrt{\frac{\bar{J}\bar{v}}{\lambda_h}}\nonumber\\
			&  +\!\big[2b_1\!\! \sum^{N - 1}_{j = \eta}\xi^{j} \! +\! \Vert \Upsilon_{L + \eta}(I)\Vert\big]\! \sqrt{\frac{(L + \eta)(N \!-\! L \!-\! \eta)\bar{J}\bar{v}}{\lambda_g}} \nonumber
			\end{align}	
			where $b_1,~ b_2,~ b_3,~ \xi >0$ are constants such that $\Vert CA^{j}\Vert \le b_1 \xi^j$, $\Vert CA^{j + \eta}\Theta^{\dag}\Vert \le b_2 \xi^j$, and $\Vert CA^{j + \eta}\Theta_{\eta}^{\dag}\Upsilon_{\eta}(I)\Vert \le b_3 \xi^j$. 
			For $n = 1,2,\cdots$ matrices $\Theta_{n}$ and $\Upsilon_{n}(I)$ are defined by   
			\begin{subequations}
				\begin{align}
				\Theta_{n} \!&\,:=\! \left[
				\begin{matrix}
				C\\
				CA\\
				\vdots\\
				CA^{n - 1}
				\end{matrix}
				\right]\! \label{eq:theta1}\\
				\Upsilon_{n}(I)\!&\,:= \!
				\left[\begin{matrix}
				\!0\! & 0 & \cdots  & 0\\
				\!C \!& 0 & \cdots & 0\\
				\!\vdots\! & \!\vdots\! & \!\ddots \!&\! \vdots\!\\
				\!CA^{n - 2}\! & CA^{n - 3} & \cdots & 0
				\end{matrix}\right]\!\label{eq:upsilonI}.
				\end{align}
			\end{subequations}	
		
	\end{lemma}
	\begin{proof}
			It follows from \eqref{eq:sys} that the inputs and outputs of system can be rewritten as
			\begin{align}\label{eq:uyuxw}
			&\quad \,\left[
			\begin{matrix}
			u_{[s_r - \eta, s_r + L - 1]}\\
			y_{[s_r - \eta, s_r + L - 1]}
			\end{matrix}
			\right] \\ 
			& = \Psi_{L + \eta}
			\left[
			\begin{matrix}
			u_{[s_r - \eta, s_r + L - 1]}\\
			x_{s_r - \eta}
			\end{matrix}
			\right] + 
			\left[
			\begin{matrix}
			0\\
			\Upsilon_{L + \eta}(I)
			\end{matrix}
			\right]
			w_{[s_r - \eta, s_r + L - 1]}
			\end{align}	 
			where
			\begin{equation}\label{eq:Psi}
			\Psi_{L + \eta} := \left[
			\begin{matrix}
			I & 0\\
			\Upsilon_{L + \eta}(B) & \Theta_{L + \eta} 
			\end{matrix}
			\right]
			\end{equation}
			and symbols $I$ and $0$ denote respectively identity and zero matrices with suitable dimensions.
			For $n = 1, 2, \cdots$, matrix $\Upsilon_{n}(B)$ is given by
			\begin{equation}
			\Upsilon_{n}(B)\!\,:= \!
			\left[\begin{matrix}
			\!0\! & 0 & \cdots  & 0\\
			\!CB \!& 0 & \cdots & 0\\
			\!\vdots\! & \!\vdots\! & \!\ddots \!&\! \vdots\!\\
			\!CA^{n - 2}B\! & CA^{n - 3}B & \cdots & 0
			\end{matrix}\right]\!\label{eq:upsilonB}.\\
			\end{equation}
			Let $x^s_N:=\{x^s_t\}_{t =0}^{N - 1}$ and $w^s_N:=\{w^s_t\}_{t =0}^{N - 1}$ denote the state trajectory and noise corresponding to the collected trajectory $(u^s_N, y^s_N)$, respectively, and one gets from \eqref{eq:uyuxw} that
			\begin{subequations}\label{eq:Huyuxw}
				\begin{align}
				\left[
				\begin{matrix}
				H_{L + \eta}(u^s_N)\\
				H_{L + \eta}(y^s_N)
				\end{matrix}
				\right] 
				& = \Psi_{L + \eta} \left[
				\begin{matrix}
				H_{L + \eta}(u^s_N)\\
				H_{1}(x^s_{N - L - \eta + 1})
				\end{matrix}
				\right]\\
				&~~~+ 
				\left[
				\begin{matrix}
				0\\
				\Upsilon_{L + \eta}(I)
				\end{matrix}
				\right]H_{L + \eta}(w^s_N).
				\end{align}	
			\end{subequations} 
			Similarly, let $(\hat{x}^s_N, u^s_N, \hat{y}^s_N)$ denote a trajectory of the system $\eqref{eq:idealsys}$, it can be derived that
			\begin{subequations}\label{eq:Huyuxw_ideal}
				\begin{align}
				\left[
				\begin{matrix}
				H_{L + \eta}(u^s_N)\\
				H_{L + \eta}(\hat{y}^s_N)
				\end{matrix}
				\right] = \Psi_{L + \eta}
				\left[
				\begin{matrix}
				H_{L + \eta}(u^s_N)\\
				H_{1}(\hat{x}^s_{ N - L - \eta+1})
				\end{matrix}
				\right].
				\end{align}
			\end{subequations}
			Choosing 
			\begin{subequations}
				\begin{align}
				x^s_0 = \hat{x}^s_0
				\end{align}
			\end{subequations}
			and substituting \eqref{eq:Huyuxw_ideal}
			into \eqref{eq:Huyuxw}, we arrive at
			\begin{align}
			&\quad \,\left[
			\begin{matrix}
			H_{L + \eta}(u^s_N) \\
			H_{L + \eta}(y^s_N)
			\end{matrix}
			\right] \nonumber\\ 
			& = \Psi_{L + \eta}
			\left[
			\begin{matrix}
			H_{L + \eta}(u^s_N)\\
			H_{1}(\hat{x}^s_{N - L - \eta+1})
			\end{matrix}
			\right] + \Psi_{L + \eta}
			\left[
			\begin{matrix}
			0\\
			W^s
			\end{matrix}
			\right] \nonumber\\
			&~~~ + 
			\left[
			\begin{matrix}
			0\\
			\Upsilon_{L + \eta}(I)
			\end{matrix}
			\right]
			H_{L + \eta}(w^s_N) \label{eq:Huyuxw_21}\\
			& = \left[\!
			\begin{matrix}
			H_{L + \eta}(u^s_N)\\
			H_{L + \eta}(\hat{y}^s_N)
			\end{matrix}\!
			\right] \!\! + \!\Psi_{L + \eta}\!\left[
			\begin{matrix}
			0\\
			W^s
			\end{matrix}
			\right]\!\! +\!\! 
			\left[
			\begin{matrix}
			0\\
			\Upsilon_{L + \eta}(I)
			\end{matrix}
			\right]\!
			H_{L + \eta}(w^s_N) \label{eq:Huyuxw_22}
			\end{align}
			where $W^s := [0, w_0^s, \cdots\!, \sum_{i = 0}^{N - L - \eta - 1}A^iw^s_{N - L - \eta - 1 - i}]$.
			Let $\left(\bar{y}^*(s_r),\bar{u}^*(s_r),g^*(s_r),h^*(s_r)\right)$ denote an optimal solution of Problem \eqref{eq:mpc} at time $s_r$.
			It follows from \eqref{eq:mpc1} that 
			\begin{align}
			\bar{y}^*(s_r) & = H_{L + \eta}(\hat{y}^s_N) g^*(s_r) - h^*(s_r)+ \Theta_{L + \eta}W^sg^*(s_r)\nonumber\\
			&~~~ + \Upsilon_{L + \eta}(I)
			H_{L + \eta}(w^s_N) g^*(s_r) \label{eq:ybar}.
			\end{align}
			Since $\hat{y}^s_N$ is an output trajectory of system \eqref{eq:idealsys}, it follows from Lemma \ref{lem:fundamental} that $H_{L + \eta}(\hat{y}^s_N) g^*(s_r)$ is a trajectory of system \eqref{eq:idealsys}.

		If we let $\tilde{y} = y - H_{L + \eta}(\hat{y}^s_N) g^*(s_r)$, and $\tilde{x}_t$ denote its corresponding state, then dynamics of $\tilde{x}_t$ is given by
			\begin{subequations}\label{eq:xtilde}
				\begin{align}
				\tilde{x}_{t + 1} &= A \tilde{x}_t + w_t\\
				\tilde{y}_t &= C\tilde{x}_t 
				\end{align}
			\end{subequations}
			with initial output condition
			\begin{align}
			&\tilde{y}_{[s_r - \eta, s_r - 1]} \label{eq:yinit} \\
			&= y_{[s_r - \eta, s_r - 1]} - [H_{\eta}(y^s_{N-L+1})g^*(s_r) \nonumber\\
			&~~~ - \Upsilon_{\eta}(I)H_{\eta}(w^s_{N-L+1})g^*(s_r) - \Theta_{\eta}W^sg^*(s_r)]\nonumber\\
			& = y_{[s_r - \eta, s_r - 1]} - [y_{[s_r - \eta, s_r - 1]} + n_{[s_r - \eta, s_r - 1]}\nonumber\\
			&~~~ + h^*_{[-\eta, -1]}(s_r)] + \Theta_{\eta}W^sg^*(s_r) \nonumber\\
			&~~~+ \Upsilon_{\eta}(I)H_{\eta}(w^s_{N-L+1})g^*(s_r) \nonumber\\
			& = - n_{[s_r - \eta, s_r - 1]}- h^*_{[-\eta, -1]}(s_r)+ \Theta_{\eta}W^sg^*(s_r)\nonumber\\
			&~~~ + \Upsilon_{\eta}(I)H_{\eta}(w^s_{N-L+1})g^*(s_r) \nonumber.
			\end{align} 
			
			Recursively, it can be obtained from \eqref{eq:xtilde} that
			\begin{equation}\label{eq:tildey}
			\tilde{y}_{[s_r - \eta, s_r - 1]} = \Theta_{\eta} \tilde{x}_{s_r - \eta}  -  \Upsilon_{\eta}(I) w_{[s_r - \eta, s_r - 1]}
			\end{equation}
			yielding
			\begin{equation}\label{eq:xtilden}
			\tilde{x}_{s_r - \eta}\! =\! \Theta_{\eta}^{\dag}\tilde{y}_{[s_r - \eta, s_r - 1]}\!-\!\Theta_{\eta}^{\dag}\Upsilon_{\eta}(I) w_{[s_r - \eta, s_r - 1]}.
			\end{equation}
			
			Based on \eqref{eq:xtilde}  and \eqref{eq:xtilden}, it follows recursively that
			\begin{align}
			\Vert \tilde{y}_{s_r + q}\Vert  & =\Big\Vert CA^{q + \eta}\tilde{x}_{s_r - \eta} + \sum_{j = 0}^{q + \eta - 1}CA^{j}w_{s_r - \eta  + j}\Big\Vert \nonumber\\
			& = \Big\Vert CA^{q + \eta}\Theta_{\eta}^{\dag}\tilde{y}_{[s_r - \eta, s_r - 1]} \nonumber\\
			&~~~- CA^{q + \eta}\Theta_{\eta}^{\dag}\Upsilon_{\eta}(I) w_{[s_r - \eta, s_r - 1]}\nonumber\\
			&~~~+ \sum_{j = 0}^{q + \eta - 1}CA^j w_{s_r  + q - j - 1} \Big\Vert\label{eq:ytilde}.
			\end{align}
		Substituting \eqref{eq:yinit} into \eqref{eq:ytilde}, one has that
			\begin{align}\label{eq:ytildeq_1}
			\Vert \tilde{y}_{s_r + q}\Vert & = \Big\Vert C A^{q + \eta}\Theta_{\eta}^{\dag}\big[\!- \!n_{[s_r - \eta, s_r - 1]}\!-\! h^*_{[-\eta, -1]}(s_r)\\
			&~~~ + \Upsilon_{\eta}(I)H_{\eta}(w^s_{N - \eta + 1})g^*(s_r) + \Theta_{\eta}W^sg^*(s_r)\big]\nonumber\\
			&~~~ - CA^{q + \eta}\Theta_{\eta}^{\dag}\Upsilon_{\eta}(I) w_{[s_r - \eta, s_r - 1]}\nonumber\\
			&~~~+ \sum_{j = 0}^{q + \eta - 1}CA^j w_{s_r  + q - j - 1}\Big\Vert \nonumber\\
			& \le b_3 \xi^{q + \eta}\sqrt{\eta}\bar{v} + b_2 \xi^{q + \eta}\big[\sqrt{\eta}\bar{v} + \Vert h^*_{[-\eta, -1]}(s_r)\Vert \nonumber\\
			&~~~+ \Vert \Upsilon_{\eta}(I) \Vert \Vert g^*(s_r)\Vert\sqrt{\eta(N - L - \eta + 1)}\bar{v}\big] \nonumber\\
			&~~~+ b_1\sqrt{\eta(N \!-\! L \!-\! \eta)}\!\!\!\! \sum^{N - L - 1}_{j = \eta}\!\!\!\!\xi^{j} \bar{v}\Vert g^*(s_r)\Vert \!+\!  b_1\!\!\!\! \sum^{q + \eta - 1}_{j = 0}\!\!\!\!\xi^{j} \bar{v}\nonumber
			\end{align}
			where constants $b_1, b_2, b_3, \xi >0$ are chosen such that for all $j \in \mathbb{N}_0$, inequalities $\Vert CA^{j}\Vert \le b_1 \xi^j$, $\Vert CA^{j + \eta}\Theta^{\dag}\Vert \le b_2 \xi^j$, and $\Vert CA^{j + \eta}\Theta_{\eta}^{\dag}\Upsilon_{\eta}(I)\Vert \le b_3 \xi^j$ hold true.
			
			Next, we derive bounds for $\Vert g^*(s_r)\Vert$ and $\Vert h^*_{[-\eta, -1]}(s_r)\Vert$.
			Noticing that $J^*(\tilde{z}_{s_r}) \le \bar{J}$, it can be deduced from \eqref{eq:mpc} that 
			\begin{align}
			\frac{\lambda_h}{\bar{v}}\Vert h^*(s_r)\Vert^2\le \bar{J} \quad &\rightarrow \quad\Vert h^*(s_r)\Vert \le \bigg(\frac{\bar{J} \bar{v}}{\lambda_h}\bigg)^{1/2}\\
			\lambda_g \bar{v}\Vert g^*(s_r)\Vert^2 \le \bar{J} \quad&\rightarrow\quad\Vert g^*(s_r)\Vert \le \bigg(\frac{\bar{J} }{\bar{v}\lambda_g}\bigg)^{1/2}.
			\end{align}
			Plugging these bounds into \eqref{eq:ytildeq_1}, leads to
			\begin{align}
			&\Vert \tilde{y}_{s_r + q}\Vert
			\le b_3 \xi^{q + \eta}\sqrt{\eta}\bar{v}+b_2 \xi^{q + \eta}\Big[\sqrt{\eta}\bar{v} + \sqrt{\bar{J}\bar{v}/\lambda_h} \nonumber\\
			&~~~+ \Vert \Upsilon_{\eta}(I) \Vert \sqrt{\eta(N - L - \eta + 1)\bar{J}\bar{v}/\lambda_g}\Big] \nonumber\\
			&~~~+ b_1\!\!\!\! \sum^{N - L - 1}_{j = \eta}\!\!\!\!\xi^{j}\sqrt{\eta(N - L - \eta)\bar{J}\bar{v}/\lambda_g} + b_1\!\!\!\! \sum^{q + \eta - 1}_{j = 0}\!\!\!\!\xi^{j} \bar{v}\nonumber\\
			& \stackrel{\triangle}{=} \beta_2(\bar{v},q)\label{eq:beta_2}
			\end{align}
			where the function $\beta_2(\bar{v},q)$ is a $\mathcal{KL}$-function.

		Finally, combining \eqref{eq:ybar} and \eqref{eq:beta_2}, we deduce that
			\begin{align*}
			&\Vert e_{y,s_r + q}\Vert  = \Vert y_{s_r + q} - \bar{y}^*_q(s_r)\big\Vert\\
			&\le \big\Vert y_{s_r + q} - I_{[q,q + n_y - 1]}H_{L + \eta}(\hat{y}^s_N)g^*(s_{r})\big\Vert\\
			&~~+ \big\Vert  I_{[q,q + n_y - 1]}H_{L + \eta}(\hat{y}^s_N)g^*(s_{r}) - \bar{y}^*_q(s_r) \big\Vert\\
			&\le \beta_2(\bar{v},q) + \big\Vert I_{[q,q + n_y - 1]}\big[- \Upsilon_{L + \eta}(I)
			H_{L + \eta}(w^s_N) g^*(s_r) \nonumber\\
			&~~ +\!h^*\!(s_r) \!-\!\Theta_{L + \eta}W^sg^*\!(s_r)\!\big]\!\big\Vert\\
			& \le \beta_2(\bar{v},q) \!+\!\! \sqrt{\frac{\bar{J}\bar{v}}{\lambda_h}} + b_1\!\!\!\! \sum^{N - 1}_{j = \eta + L}\!\!\!\!\xi^{j}\sqrt{\frac{(N \!-\! L\! - \eta)\bar{J}\bar{v}}{\lambda_g}} \\
			&~~+ \Vert \Upsilon_{L + \eta}(I)\Vert \sqrt{\!{(L\! +\! \eta)(N \!- \!L \!- \eta + 1)\bar{J}\bar{v}}/{\lambda_g}}\\
			& \le \beta_1(\bar{v},q)
			\end{align*}
			holds for all $q \!\in\! \{0, \cdots\!, L - 1\}$, where the $\mathcal{KL}$-function $\beta_1(\bar{v},q)$ is defined in \eqref{eq:beta_1}.
			
			So far, we have bounded the error between the predicted output by using the controller in Algorithm \ref{alg:mpc} and the actual output for the case of $L \ge T$.
			When $L < T$, it follows from Algorithm \ref{alg:mpc} that zero inputs will be used once the $L$ predicted control inputs are run out.
			Notice from \eqref{eq:mpc3} that the last $\eta$ predicted inputs are all zeros too.
			Hence, the established bound in \eqref{eq:boundey} still holds for all $t \in [s_r + L, s_{r+1})$. Details of the proof for the case of $L < T$ are thus omitted.
	\end{proof}
	\begin{remark}[{\emph{Trade-off between level of noise and DoS}}]\label{rmk:tradeoff}
		According to \eqref{eq:beta_1},
		the error between the estimated output $\bar{y}^*_{q}(s_r)$ and the actual output $y_{s_r + q}$ increases as $q$ grows, and reset to a constant related to the noise bound at every successful transmission (i.e., $q = 0$).
		In addition, Lemma \ref{lem:dos} indicates stronger DoS attacks lead to longer $T \ge s_{r + 1} - s_{r}$, and hence larger $q$.
		Therefore, it suggests that, to stabilize systems under stronger DoS attacks, the level of noise $w_t$ and $n_t$ should be smaller.
		Therefore, there exists a trade-off between robustness against noise as well as resilience against DoS attacks.
	\end{remark}
	
	Having bounded the prediction error, we analyze the closed-loop stability of system \eqref{eq:sys} along with the proposed controller in Algorithm \ref{alg:mpc}. To this end, let us begin by constructing a Lyapunov function.
	Based on the observability of state $x$ and Definition \ref{pro:uioss}, there exists a positive definite matrix $P = P^\top \succ 0$ such that $W(z) = z^\top P z$ is a UIOSS-Lyapunov function.
	Consider the following Lyapunov function accounting for both the data-driven MPC solution and the augmented system state
	\begin{equation}\label{eq:lapfunction}
	V_t := J^*_L(\tilde{z}_t) + \gamma W(z_t)
	\end{equation}
	where $\gamma>0$ is any positive constant. 
	The next lemma provides lower and upper bounds for this Lyapunov function in terms of the augmented state $z_t$.

	\begin{lemma}[{\emph{Bounds on the Lyapunov function}}]\label{lem:boundV}
		Consider the system \eqref{eq:sys} with the controller in Algorithm \eqref{alg:mpc}.
		Suppose that Assumptions \ref{as:ctrl}--\ref{as:horizon} are met.
		If the parameters of DoS attacks satisfying  \eqref{eq:doscondition},
		then there exist constant $\delta > 0$ such that for all $z_{s_r} \in \mathbb{B}_{\delta}$, Problem \eqref{eq:mpc} is feasible, and the Lyapunov function $V_{s_r}$ is bounded for $r \in \mathbb{N}_{0}$ as follows
			\begin{equation}\label{eq:boundV}
			\gamma\underline{\lambda}_{P}\Vert z_{s_r}\Vert^2 \le V_{s_r} \le \gamma_3 \Vert z_{s_r}\Vert^2 + \alpha_3(\bar{v})
			\end{equation}
			where $\gamma_3$ is a constant, and $\alpha_3(\bar{v})$ is a $\mathcal{K}_{\infty}$-function.
	\end{lemma}
	
	\begin{proof}
			Consider an arbitrary $s_r \in\mathbb{N}_0$ with $r \in \mathbb{N}_{0}$.
			Since $J^*_L(\tilde{z}_{s_r}) >0$, it follows from \eqref{eq:lapfunction} that 
			$V_{s_r} \ge \gamma W(z_{s_r})\ge  \gamma \underline{\lambda}_{P}\Vert z_{s_r}\Vert^2$.
			It remains to establish the upper bound in \eqref{eq:boundV}.
			
			To this end, let us start by constructing a candidate solution $(\bar{g}(s_r), \bar{h}(s_r), \bar{u}(s_r), \bar{y}(s_r))$ of Problem \eqref{eq:mpc}, in which, the input sequence is expected to bring the  state $x$ (and its corresponding output $y$) of system \eqref{eq:sys} to a ball around the origin whose size depends on the magnitude of noise $\bar{v}$ in $L$ steps.
			In the following, we begin by choosing $\bar{u}(s_r)$ and $\bar{y}(s_r)$, which is followed by proving the existence of $\bar{g}(s_r)$ and $\bar{h}(s_r)$ such that constraints in \eqref{eq:mpc} are satisfied.
			
			According to \eqref{eq:mpc2}, it holds that $\bar{u}_{[-\eta, -1]}(s_r) = u_{[s_r - \eta, s_r - 1]}$ and $\bar{y}_{[-\eta, -1]}(s_r) = {\zeta}_{[s_r - \eta, s_r - 1]}$.
			Capitalizing on the system stabilizability, for the ideal system \eqref{eq:idealsys}, there exist a constant $\delta > 0$ and an input sequence $u_{[s_r, s_r + L- 1]} \in \mathbb{U}^L$, such that for all $x_{s_r}$ satisfying $\Vert x_{s_r}\Vert/\gamma_z \le \Vert z_{s_r}\Vert \le \delta$, this input sequence can bring the state
			$x_{[s_r, s_r + L - 1]}$ 
			and associated output
			$\hat{y}_{[s_r, s_r + L - 1]}$ 
			to the origin in $L - \eta$ steps (due to constraint \eqref{eq:mpc3}, there are actually at most $L - \eta$ non-zero inputs in the solution $\bar{u}_{[s_r, s_r + L - 1]}$), while obeying
			\begin{equation}\label{eq:gammauy}
			\left\Vert
			\left[
			\begin{matrix}
			u_{[s_r, s_r + L - 1]}\\
			\hat{y}_{[s_r, s_r + L - 1]}
			\end{matrix} \right]
			\right\Vert^2
			\le \gamma_{uy}^2\Vert x_{s_r}\Vert^2.
			\end{equation}	
			Hence, we choose $\bar{u}_{[0, L - 1]}(s_r) = u_{[s_r, s_r + L - 1]}$ and $\bar{y}_{[0, L - 1]}(s_r) = \hat{y}_{[s_r, s_r + L - 1]}$.
			Next, we prove that there exists a vector $g(s_r)$ such that \eqref{eq:mpc1} is satisfied.
			
			Set  
			\begin{equation*}
			\bar{g}(s_r) = H^{\ddag}_{ux}
			\left[
			\begin{matrix}
			u_{[s_r - \eta, s_r + L - 1]}\\
			x_{s_r - \eta}
			\end{matrix}
			\right]
			\end{equation*}
			with matrix $H_{ux}$ defined by
			\begin{align*}\label{eq:Hux}
			&H_{ux} := \left[
			\begin{matrix}
			H_{L + \eta}(u^s_N)\\
			H_{1}(\hat{x}^s_{N - L - \eta + 1 })
			\end{matrix}
			\right].
			\end{align*}	 
			Based on \eqref{eq:Huyuxw_ideal}, 
			equation \eqref{eq:mpc1} can be rewritten as
			\begin{subequations}
				\begin{align}
				&\left[\!
				\begin{matrix}
				{H}_{L + \eta}(u^s_N)\\
				{H}_{L + \eta}(y^s_N)
				\end{matrix}
				\!\right] \!\bar{g}(s_r)  = \Psi_{L + \eta}H_{ux} H^{\ddag}_{ux}
				\left[\!
				\begin{matrix}
				u_{[s_r - \eta, s_r + L - 1]}\\
				x_{s_r - \eta}
				\end{matrix}
				\!\right] \\
				&~~~+ 
				\bigg(\!\!\Psi_{L + \eta}\! 
				\left[\!\begin{matrix}
				0\\
				W^s
				\end{matrix}\!\right]\!+\! \left[\!\begin{matrix}
				0\\
				\Upsilon_{L + \eta}
				\end{matrix}\!\right]\!H_{L + \eta}(w^s)\!\!\bigg)\bar{g}(s_r)\nonumber\\
				& = \left[\!
				\begin{matrix}
				u_{[s_r - \eta, s_r + L - 1]}\\
				\hat{y}_{[s_r - \eta, s_r + L - 1]}\end{matrix}\!
				\right] \!\!+ \!
				\bigg(\!\!\Psi_{L + \eta}\! 
				\left[\!\begin{matrix}
				0\\
				W^s
				\end{matrix}\!\right]\!+\! \left[\!\begin{matrix}
				0\\
				\Upsilon_{L + \eta}
				\end{matrix}\!\right]\!H_{L + \eta}(w^s)\!\!\bigg)\bar{g}(s_r)\nonumber\\
				& = \left[
				\begin{matrix}
				u_{[s_r - \eta, s_r + L - 1]}\\
				\zeta_{[s_r - \eta, s_r - 1]}\\
				\hat{y}_{[s_r, s_r + L - 1]}
				\end{matrix}
				\right]
				\!+ \!
				\left[
				\begin{matrix}
				0\\
				\bar{h}_{[- \eta, - 1]}(s_r)\\
				\bar{h}_{[0, L - 1]}(s_r)
				\end{matrix}
				\right].
				\end{align} 
			\end{subequations}
			Combining \eqref{eq:sys} and \eqref{eq:idealsys} yields
			\begin{align*}
			\zeta_{[s_r - \eta, s_r - 1]} & = \hat{y}_{[s_r - \eta, s_r - 1]}\\
			&~~~ +
			\Upsilon_{\eta}(I)w_{[s_r - \eta, s_r - 1]}  + n_{[s_r - \eta, s_r - 1]}.
			\end{align*} 
			Therefore, we choose the candidate $\bar{h}(s_r)$ as follows
			\begin{subequations}\label{eq:hbar}
				\begin{align}
				\bar{h}_{[-\eta, -1]}(s_r) &\!=\! [I, 0]\!\big(\Theta_{L + \eta} W^s - \Upsilon_{\eta}(I)w_{s_r - \eta, s_r - 1} \\
				&~~+ \Upsilon_{L + \eta}(I)H_{L + \eta}(w^s)\big)\bar{g}(s_r) - n_{[s_r - \eta, s_r - 1]}\nonumber\\
				\bar{h}_{[0, L - 1]}(s_r) & \!=\! [0, I]\!\big(\Theta_{L + \eta} W^s + \Upsilon_{L + \eta}(I)H_{L + \eta}(w^s)\big)\bar{g}(s_r).
				\end{align}
		\end{subequations}
		In addition, \eqref{eq:hbar} indicates that 
			\begin{align}\label{eq:hbargbar}
			\Vert \bar{h}(s_r)\Vert 
			& \le \bar{v}\Big[\Big(b_1\!\!\! \!\sum_{j = L + \eta}^{N -1}\!\!\!\!\xi^j \sqrt{(L + \eta)(N - L - \eta)} \nonumber\\
			&~~~+ \|\Upsilon_{L + \eta}(I)\|\sqrt{(L + \eta)(N - L - \eta + 1)}\Big)\|\bar{g}(s_r)\| \nonumber\\
			&~~~ + (1 + \|\Upsilon_{\eta}(I)\|)\sqrt{\eta}\Big].
			\end{align}
			Noticing from \eqref{eq:mpc} that the coefficients for terms $\Vert \bar{h}(s_r)\Vert$ and $\Vert \bar{g}(s_r)\Vert$ are $\lambda_h/\bar{v}$ and  $\lambda_g\bar{v}$, respectively.
			Therefore, to minimize the cost function $J^*_L$, the inequality \eqref{eq:hbargbar} is naturally guaranteed, and a candidate solution of Problem \eqref{eq:mpc} is constructed.
			
			Now, we are ready to upper bound $V_{s_r}$ in terms of $z_{s_r}$ and $\bar{v}$.
			Since 
			\begin{equation*}
			\left[
			\begin{matrix}
			u_{[s_r - \eta, s_r - 1]}\\
			y_{[s_r - \eta, s_r - 1]}
			\end{matrix}
			\right]  = \Psi_{\eta}
			\left[
			\begin{matrix}
			u_{[s_r - \eta, s_r - 1]}\\
			x_{s_r - \eta}
			\end{matrix}
			\right] + 
			\left[
			\begin{matrix}
			0\\
			\Upsilon_{\eta}(I)
			\end{matrix}
			\right]
			w_{[s_r - \eta, s_r - 1]}
			\end{equation*}	 
			from the observability of system \eqref{eq:sys} we have that
			\begin{equation}\label{eq:uxzw}
			\left[
			\begin{matrix}
			u_{[s_r - \eta, s_r - 1]}\\
			x_{s_r - \eta}
			\end{matrix}
			\right] = \Psi^{\dag}_{\eta}z_{s_r} - \Psi^{\dag}_{\eta}\left[
			\begin{matrix}
			0\\
			\Upsilon_{\eta}(I)
			\end{matrix}
			\right]
			w_{[s_r - \eta, s_r - 1]}.
			\end{equation}
			Therefore, $\|\bar{g}(s_r)\|$ satisfies
			\begin{align*}
			\|\bar{g}(s_r)\| & \le \|H_{ux}^{\ddag}\|^2 \bigg(\!\|u_{[s_r, s_r + L - 1]}\|^2 + \left\|\left[
			\begin{matrix}
			u_{[s_r - \eta, s_r - 1]}\\
			x_{s_r - \eta}
			\end{matrix}
			\right] \right\|^2\!\bigg) \nonumber\\
			& \le \|H_{ux}^{\ddag}\|^2 (\gamma_{uy}^2 \|x_{s_    r}\|^2 + 2\|\Psi^{\dag}_{\eta}\|^2\|z_{s_r}\|^2 \\
			&~~~+ 2\|\Theta^\dag_{\eta}\Upsilon_{\eta}(I)\|^2 \eta\bar{v}^2)\\
			& \le \gamma_1\|z_{s_r}\|^2 + \gamma_2 \bar{v}^2
			\end{align*}
			where $\gamma_1 := \|H_{ux}^\ddag\|^2(\gamma_{uy}^2 \gamma_z^2 + 2\|\Psi^{\dag}_{\eta}\|^2)$ and $\gamma_2 := 2\|H_{ux}^\ddag\|^2 \|\Theta^\dag_{\eta}\Upsilon_{\eta}(I)\|^2\eta$.
			The first inequality is derived using the definition of $2$-norm.
			Substituting \eqref{eq:gammauy} and \eqref{eq:uxzw} into the first inequality, we arrive at the second inequality.  
			It follows from \eqref{eq:hbargbar} that
			\begin{align*}
			\|\bar{h}(s_{r})\|^2 &\le 2\Bigg[2\Bigg(b_1\!\!\! \!\sum_{j = L + \eta}^{N -1}\!\!\!\!\xi^j\Bigg)^2 (L + \eta)(N - L - \eta) \\
			&~~~+ \!2 \Vert \Upsilon_{L + \eta}(I)\Vert^2(L  \!+\! \eta)\!(N \!-\! L \!-\! \eta \!+ \!1)\!\Bigg]\!\bar{v}^2 \Vert \bar{g}(s_{r})\Vert^2 \\
			&~~~+ \!2\Vert \Upsilon_{\eta}(I)\Vert^2 \eta\bar{v}^2 + 2 \eta\bar{v}^2\\
			& \le \alpha_1(\bar{v}^2)\gamma_1 \Vert z_{s_r}\Vert^2 + \alpha_2(\bar{v}^2)
			\end{align*}
			where $\alpha_1(\bar{v}^2)$ and $\alpha_2(\bar{v}^2)$ are $\mathcal{K}_{\infty}$-functions defined by $\alpha_1(\bar{v}^2) := 2[2(b_1\sum_{j = L + \eta}^{N -1}\xi^j)^2 (L + \eta)(N - L - \eta) + 2 \Vert \Upsilon_{L + \eta}(I)\Vert^2(L  + \eta)(N - L - \eta - 1)]\bar{v}^2$ and $\alpha_2(\bar{v}^2) := \alpha_1(\bar{v}^2)\bar{v}^2 \gamma_2 + 2\Vert \Upsilon_{\eta}(I)\Vert^2 \eta\bar{v}^2 + 2 \eta\bar{v}^2$.
			Recalling $V_t := J^*_L(\tilde{z}_t) + \gamma W(z_t)$, it can be deduced that
			\begin{align*}
			V_{s_r} &\le \overline{\lambda}_{\{R_1, R_2\}} \gamma_{uy}^2\gamma_{z}^2\Vert z_{s_r}\Vert^2 +  \gamma\overline{\lambda}_{P}\Vert z_{s_r}\Vert^2 + \lambda_g \bar{v} (\gamma_1 \Vert z_{s_r}\Vert^2 \\
			&~~~+ \gamma_2 \bar{v}^2) + (\lambda_h/\bar{v})(\alpha_1(\bar{v}^2) \gamma_1 \Vert z_{s_r}\Vert^2 + \alpha_2(\bar{v}^2))\\
			& \le \gamma_3 \Vert z_{s_r}\Vert^2 + \alpha_3(\bar{v})
			\end{align*}
			where $\gamma_3 := \overline{\lambda}_{\{R_1, R_2\}}\gamma_{uy}^2\gamma_{z}^2 + \gamma \overline{\lambda}_{P} + \lambda_g \gamma_1 \bar{v} + \lambda_h\gamma_1 \alpha_1(\bar{v}^2)/\bar{v}$, and $\alpha_3(\bar{v}) := \lambda_g \gamma_2\bar{v}^3 + (\lambda_h/\bar{v})\alpha_2(\bar{v}^2)$ is a $\mathcal{K}_{\infty}$-function.
	\end{proof}

	Now, based on Lemmas \ref{lem:bounde} and \ref{lem:boundV}, our stability result is presented as follows.
	\begin{theorem}\label{thm:isps}
		Consider the system \eqref{eq:sys} adopting the controller in Algorithm \ref{alg:mpc}, and there exist a constant $\delta >0$ such that  the initial condition $z_0 \!\in\! \mathbb{B}_{\delta}$.
			Let Assumptions \ref{as:ctrl}--\ref{as:horizon} hold.
			For any $V_{in} > 0$, there exist positive constants $\overline{\lambda}_g$, $\underline{\lambda}_g$, $\overline{\lambda}_h$, $\underline{\lambda}_h$, and $\bar{v}_0 >0$ such that for all $\underline{\lambda}_g \!\le\! \lambda_g \!\le\!\overline{\lambda}_g$, $\underline{\lambda}_h \!\le\! \lambda_h \!\le\!\overline{\lambda}_h$, and $\bar{v} \le \bar{v}_0$, 
		if i) the DoS attacks obey \eqref{eq:doscondition} in Lemma \ref{lem:dos}, ii) Problem \eqref{eq:mpc} is feasible at $s_{0}$, and the Lyapunov function satisfies $V_{s_0} \le V_{in}$, then the following statements hold true
		\begin{enumerate}
			\item Problem \eqref{eq:mpc} is feasible at any successful transmission instant $s_r$ with $ r \in \mathbb{N}_0$; and,
			\item System \eqref{eq:sys} is ISS under Algorithm \ref{alg:mpc}.
		\end{enumerate}
	\end{theorem}
	\begin{proof}
		We begin  by proving 1), which is also known as \emph{Recursive feasibility}.
		
		The proof consists of two steps.
			In the first step, we show that if Problem \eqref{eq:mpc} is feasible at $s_0$, it is feasible at $s_{1}$. 
			In the second step, we prove that the Lyapunov function decreases at every successful transmission time instant, which indicates the recursive feasibility of Problem \eqref{eq:mpc}.
		
		\textbf{Step 1}: First, we prove that Problem \eqref{eq:mpc} is feasible at $s_1$, by means of carefully constructing a candidate solution.
		Denote an optimal solution of the Problem \eqref{eq:mpc} at time instant $s_0$ by $(g^*(s_0), h^*(s_0), \bar{u}^*(s_0), \bar{y}^*(s_0))$.
		The true trajectory of system \eqref{eq:sys} within the time interval $[s_0, s_1)$ with initial condition $(u_{[s_0 - \eta, s_0 - 1]}, \zeta_{[s_0 - \eta, s_0 - 1]})$ and the input sequence $\bar{u}^*(s_0)$ is represented by $y_{[s_0, s_1]}$.
		Denote a candidate solution at $s_1$ by $(\bar{g}(s_1), \bar{h}(s_1), \bar{u}(s_1), \bar{y}_{s_1})$.
		The following proof is divided into three cases with respect to the relationship between $s_1 - s_0$ and $n_x$. 
		
		\textbf{Case 1:} $s_1 - s_0 < n_x$.
		
		Similar to the proof of Lemma \ref{lem:boundV}, we first construct $\bar{u}$, $\bar{y}$, and then $\bar{g}$ and $\bar{h}$ can be derived using the constraints in Problem \eqref{eq:mpc}.
		For $q \in [- \eta, L - \eta - n_x - 1]$, let $\bar{u}_{q}(s_1) = \bar{u}^*_{q + s_1 - s_0}(s_0)$.
		For $q \in [-\eta, -1]$, let $\bar{y}_q(s_1) = \zeta_{s_1 + q}$, and for $q \in [0, L - \eta - n_x - 1]$, let $\bar{y}_{q}(s_1) = y_{s_1 + q}$.
		Recalling from Lemma \ref{lem:bounde} that for $q \in [L - \eta - n_x, L - \eta - 1]$
		\begin{equation}
		\Vert y_{q}(s_1) - \bar{y}^*_{q + s_1 - s_0}(s_0)\Vert \le \beta_1(\bar{v}, q).
		\end{equation}
		In addition, since $x_{s_0}\!\in\! \mathbb{B}_{\delta/\gamma_z}$, and the constraint \eqref{eq:mpc3} $(\bar{y}^*_{[L - \eta, L - 1]}(s_0),\bar{u}^*_{[L - \eta, L - 1]}(s_0))$ equals to zero, it can be deduced that $\bar{y}^*_{[q + s_1 - s_0, q + s_1 - s_0 + n_x]}(s_0)$ is close to zero when $\bar{v}$ approaches to zero (if $\bar{y}^*_{[q + s_1 - s_0, q + s_1 - s_0 + n_x]}(s_0)$ is large, then constraint \eqref{eq:mpc3} is violated).
		Consequently, if $\bar{v} \rightarrow 0$, the output $y_{[L - \eta - n_x, L - \eta - 1]}(s_1)$ and its corresponding state $x_{L - \eta - n_x}(s_1)$ approach to zero.
		In this setting, since the system is stabilizable, similar to \eqref{eq:gammauy}, there exist an input trajectory $\bar{u}_{[L - \eta - n_x, L - \eta - 1]}(s_1)$ that brings the state and its corresponding output of the ideal system \eqref{eq:idealsys} $\hat{y}$ to zero while obeying 
		\begin{equation}
		\left\Vert
		\left[
		\begin{matrix}
		\bar{u}_{[L - \eta - n_x, L - \eta - 1]}(s_1)\\
		\hat{y}_{[s_1 + L - \eta - n_x, s_1 + L - \eta - 1]}
		\end{matrix} \right]
		\right\Vert^2
		\le \gamma_{uy}^2\Vert x_{s_1 + L - \eta - n_x}\Vert^2.
		\end{equation}
		Let $\bar{y}_q(s_1) = \hat{y}_{s_1 + q}$ for $q \in [L - \eta - n_x, L - \eta - 1]$.
		Moreover, due to constraint \eqref{eq:mpc3}, for $q \in [L - \eta, L - 1]$, we let $\bar{y}_{[L - \eta, L - 1]}(s_1)$ and	  $\bar{u}_{[L - \eta, L - 1]}(s_1)$ equal to zero.
		
		Based on this input-output pair, $\bar{g}(s_1)$ and $\bar{h}(s_1)$ can be chosen following the same step from that of Lemma \eqref{lem:boundV}, i.e.,
		\begin{equation}
		\bar{g}(s_1) = H^{\ddag}_{ux}
		\left[
		\begin{matrix}
		u_{[s_1 - \eta, s_1 + L - 1]}\\
		x_{s_1 - \eta}
		\end{matrix}
		\right].
		\end{equation}
		and 
		\begin{align*}
		&\bar{h}_{[-\eta, -1]}(s_1) \!=\! -\Upsilon_{\eta}w_{[s_1 - \eta, s_1 - 1]} - n_{[s_1 - \eta, s_1 - 1]} - \Psi^{w} \bar{g}(s_1)\\
		&\bar{h}_{[0, L - \eta - n_x - 1]}(s_1) = -\Upsilon_{L - \eta - n_x}w_{[s_1, s_1 + L - \eta - n_x - 1]} - \Psi^{w} \bar{g}(s_1)\\
		&\bar{h}_{[L - \eta - n_x, L - 1]}(s_1)=- \Psi^{w} \bar{g}(s_1)  + \Theta_{n_x}\!\!\!\!\!\sum_{i = 0}^{L - n_x - 1}\!\!\!\! A^i w_{s_1 + L - \eta - n_x - i - 1}
		\end{align*}
		where
		\begin{equation*}
		\Psi^{w} := \bigg(\!\!\Psi_{L + \eta}\! 
		\left[\!\begin{matrix}
		0\\
		W^s
		\end{matrix}\!\right]\!+\! \left[\!\begin{matrix}
		0\\
		\Upsilon_{L + \eta}
		\end{matrix}\!\right]\!H_{L + \eta}(w^s_N)\!\!\bigg)
		\end{equation*}
		
		\textbf{Case 2:} $n_x \le s_1 - s_0  < L$.
		
		For $q \in [-\eta, L - \eta - (s_1 - s_0) - 1]$, let $\bar{u}_q(s_1) = \bar{u}^*_{q + s_1 - s_0}(s_0)$.
		For $q \in [-\eta, -1]$, let $\bar{y}_q(s_1) = \zeta_{s_1 + q}$, and for $q \in [0, L - \eta - (s_1 - s_0) - 1]$, let $\bar{y}_q(s_1) = y_{s_1 + q}$.
		Since $s_1 - s_0 \ge n_x$, it follows from \eqref{eq:mpc3} that for $q \in [L - \eta - (s_1 - s_0), L - 1 - (s_1 - s_0)]$, $\bar{y}^*_{q + s_1 - s_0}(s_0) = 0$.
		Based on Lemma \ref{lem:bounde}
		\begin{equation}
		\Vert y_{q}(s_1)\Vert \le \beta_1(\bar{v},q)
		\end{equation}
		and hence the state $x_{s_0 + L - \eta}$ is close to zero when $\bar{v}$ approaches to zero.
		In this setting, similar to that of Case(1), there exist an input trajectory $\bar{u}_{[L - \eta - (s_1 - s_0)], L - \eta - 1]}(s_1)$ that brings the state and its corresponding output of the ideal system \eqref{eq:idealsys} $\hat{y}$ to zero while obeying 
		\begin{equation}\label{eq:gammauy2}
		\left\Vert\!
		\left[
		\begin{matrix}
		\bar{u}_{[L - \eta - (s_1 - s_0), L - \eta - 1]}(s_1)\\
		\hat{y}_{[s_0 + L - \eta, s_1 + L - \eta - 1]}
		\end{matrix} \right]\!
		\right\Vert^2
		\le \gamma_{uy}^2\Vert x_{s_0 + L - \eta }\Vert^2.
		\end{equation}
		Let $\bar{y}_q(s_1) = \hat{y}_{s_1 + q}$ for $q \in [L - \eta - (s_1 - s_0), L - \eta - 1]$.
		Moreover, due to constraint \eqref{eq:mpc3}, for $q \in [L - \eta, L - 1]$, we let $\bar{y}_{[L - \eta, L - 1]}(s_1)$ and	  $\bar{y}_{[L - \eta, L - 1]}(s_1)$ equal to zero.
		
		Based on this input-output pair, $\bar{g}(s_1)$ and $\bar{h}(s_1)$ can be chosen as,
		\begin{equation}\label{eq:gbars1}
		\bar{g}(s_1) = H^{\ddag}_{ux}
		\left[
		\begin{matrix}
		u_{[s_1 - \eta, s_1 + L - 1]}\\
		x_{s_1 - \eta}
		\end{matrix}
		\right]
		\end{equation}
		and 
		\begin{subequations}
			\begin{align}
			&\bar{h}_{[-\eta, -1]}(s_1) \nonumber\\
			&~~~= -\Upsilon_{\eta}\!w_{[s_1 - \eta, s_1 - 1]} - n_{[s_1 - \eta, s_1 - 1]} - \Psi^{w} \bar{g}(s_1)\label{eq:hbar1}\\
			&\bar{h}_{[0, L - \eta - (s_1 - s_0) - 1]}(s_1) \nonumber\\
			&~~~=- \Psi^{w} \bar{g}(s_1)-\Upsilon_{L - \eta - (s_1 - s_0)}w_{[s_1, s_0  + L - \eta - 1]} \label{eq:hbar12}\\
			&\bar{h}_{[L - \eta - (s_1 - s_0), L - 1]}(s_1)   \nonumber\\
			&~~~=- \Psi^{w} \bar{g}(s_1)+\Theta_{(s_1 - s_0)}\!\!\!\!\!\!\!\!\!\!\!\sum_{i = 0}^{L - (s_1 - s_0) - 1} \!\!\!\!\!\!\!\!\!\!\!A^i w_{s_0 + L - \eta - i - 1}\label{eq:hbar3}.
			\end{align}	
		\end{subequations}
		
		\textbf{Case 3:} $s_1 - s_0 \ge L$.
		
		According to Lemma \ref{lem:bounde}, the error bound between the predicted output and the actual output still holds true when $s_1 - s_0 \ge L$.
		Therefore, a candidate solution can be constructed follows the same step as in the Case (2) and is omitted here.
		
		\textbf{Step 2:} Next, leveraging the candidate solution constructed above, the recursive feasibility is proved by showing that the Lyapunov function decreases at every successful transmission instant. 
		For simplicity, we only prove for Case (2), and the proof for Case (1) and (3) carry through.
		Suppose that $\bar{v}_0$ is sufficiently small such that candidate solution in Step 1 can be constructed.
		The optimal cost obeys
		\begin{align}
		&J^*_L(\tilde{z}(s_1))
		\le J_{L}(\tilde{z}(s_1), \bar{g}(s_1), \bar{h}(s_1))\\
		& = \bigg( J^*_L(\tilde{z}(s_0)) - \sum_{q = 0}^{L - 1}\ell(\bar{u}_q^*(s_0), \bar{y}_q^*(s_0)) - \lambda_g \bar{v} \Vert g^*(s_0)\Vert^2\nonumber\\
		&~~~ - (\lambda_h/\bar{v}) \Vert h^*(s_0)\Vert^2 \bigg)+ \lambda_g \bar{v} \Vert \bar{g}(s_1)\Vert^2 + (\lambda_h/\bar{v}) \Vert \bar{h}(s_1)\Vert^2\nonumber\\
		&~~~+\sum_{q = 0}^{L - 1}\ell(\bar{u}_q(s_1), \bar{y}_q(s_1))\\
		& \le J^*_L(\tilde{z}(s_0)) + (\lambda_h/\bar{v}) \Vert \bar{h}(s_1)\Vert^2 + \lambda_g \bar{v} \Vert \bar{g}(s_1)\Vert^2 
		\nonumber\\
		&~~~+ \sum_{q = 0}^{L - 1}\big(\ell(\bar{u}_q(s_1), \bar{y}_q(s_1)) - \ell(\bar{u}_q^*(s_0), \bar{y}_q^*(s_0)) \big)\label{eq:ls1ls0}.
		\end{align}
		In addition, the last term in \eqref{eq:ls1ls0} can be decomposed by
		\begin{subequations}\label{eq:ll}
			\begin{align}
			&\sum_{q = 0}^{L - 1}\ell(\bar{u}_q^*(s_1), \bar{y}_q^*(s_1)) - \sum_{q = 0}^{L - 1}\ell(\bar{u}_q(s_0), \bar{y}_q(s_0)) \\
			& = - \!\!\!\!\!\!\sum_{q = 0}^{s_1 - s_0 - 1}\!\!\!\!\ell(\bar{u}_q(s_0), \bar{y}_q(s_0)) + \!\!\!\!\!\!\!\!\!\!\!\!\sum_{q = L - \eta - (s_1 - s_0)}^{L - \eta - 1}\!\!\!\!\!\!\!\!\!\!\ell(\bar{u}_q^*(s_1), \bar{y}_q^*(s_1))\label{eq:ll2}\\
			&~~~ +\!\!\!\!\!\!\!\!\!\!\!\!\!\!\!\! \sum_{q = 0}^{L - \eta - (s_1 - s_0) - 1}\!\!\!\!\!\!\!\!\!\!\!\!\!\!\!\big(\ell(\bar{u}_q(s_1), \bar{y}_q(s_1)) \!-\! \ell(\bar{u}_{q + s_1 - s_0}^*\!(s_0),\! \bar{y}_{q + s_1 - s_0}^*\!(s_0))\big)\label{eq:ll3}.
			\end{align}
		\end{subequations}
		For the second term in \eqref{eq:ll2}, since $\bar{y}_i(s_1)$ is constructed to be the output for the ideal system \eqref{eq:idealsys}, the observability of the system indicates that
		\begin{equation}\label{eq:xyhat}
		\Vert x_{s_0 + L - \eta}\Vert^2 \le \Vert \Theta_\eta^{\dag} \Vert^2 \Vert \hat{y}_{[s_0 + L - \eta, s_0 + L- 1]}\Vert^2.
		\end{equation} 
		In addition, noticing that $\bar{y}^{*}_{[L - \eta, L - 1]}(s_0) = 0$, similar to Lemma \ref{lem:bounde}, there exist a function $\alpha_4 \in \mathcal{K}_{\infty}$ such that
		\begin{equation}\label{eq:alpha4}
		\Vert \hat{y}_{[s_0 + L - \eta, s_0 + L- 1]}\Vert \le \alpha_4(\bar{v}).
		\end{equation}
		Therefore, leveraging \eqref{eq:gammauy2}, we arrive at
		\begin{align}\label{eq:l1}
		\sum_{q = L - \eta - (s_1 - s_0)}^{L - \eta - 1}\!\!\!\!\!\!\!\!\!\!\!\!\!\!\ell(\bar{u}_q(s_1), \bar{y}_q(s_1))\! \le\! \overline{\lambda}_{\{R_1, R_2\}}\gamma_{uy}^2\Vert \Theta_\eta^{\dag} \Vert^2 \alpha_4^2(\bar{v})
		\end{align}
		which approaches to zero as $\bar{v} \rightarrow 0$.
		
		Next, we derive an upper bound for the term in \eqref{eq:ll3}.
		Recalling that $\bar{u}^*_{q + s_1 - s_0}(s_0) = \bar{u}_{q}(s_1)$ for $q \in [0, L - \eta - (s_1 - s_0) - 1]$, one gets that
		\begin{align}
		\eqref{eq:ll3} =\!\!\!\!\!\!\!\!\!\!\!\! \sum_{q = 0}^{L - \eta - (s_1 - s_0) - 1}\!\!\!\!\!\!\!\!\!\!\!\!\big(\Vert \bar{y}_q(s_1) \Vert^2_{R_2}- \Vert \bar{y}_{q + s_1 - s_0}^*(s_0)\Vert^2_{R_2} \big).
		\end{align}
		In addition,
		\begin{align}
		&\Vert \bar{y}_q(s_1) \Vert^2_{R_2}- \Vert \bar{y}_{q + s_1 - s_0}^*(s_0)\Vert^2_{R_2}\nonumber\\ &\le \overline{\lambda}_{R_2} \big(\Vert \bar{y}_q(s_1) \Vert^2- \Vert \bar{y}_{q + s_1 - s_0}^*(s_0)\Vert^2\big)\nonumber\\
		&= \overline{\lambda}_{R_2} \big(\Vert \bar{y}_q(s_1) - \bar{y}_{q + s_1 - s_0}^*(s_0) + \bar{y}_{q + s_1 - s_0}^*(s_0) \Vert^2\nonumber\\
		&~~~- \Vert \bar{y}_{q + s_1 - s_0}^*(s_0)\Vert^2\big)\nonumber\\
		&\le \overline{\lambda}_{R_2} \big(\Vert \bar{y}_q(s_1) - \bar{y}_{q + s_1 - s_0}^*(s_0) \Vert^2 \nonumber\\
		&~~~+ 2\Vert \bar{y}_q(s_1) - \bar{y}_{q + s_1 - s_0}^*(s_0) \Vert \Vert \bar{y}_{q + s_1 - s_0}^*(s_0)\Vert\big)\nonumber\\
		& \le  \overline{\lambda}_{R_2} \big(\Vert \bar{y}_q(s_1) - \bar{y}_{q + s_1 - s_0}^*(s_0) \Vert^2 \nonumber\\
		&~~~+ 2\Vert \bar{y}_q(s_1) - \bar{y}_{q + s_1 - s_0}^*(s_0) \Vert (1 +  \Vert \bar{y}_{q + s_1 - s_0}^*(s_0)\Vert)\big)\nonumber\\
		& \le \overline{\lambda}_{R_2} \big(\Vert \bar{y}_q(s_1) - \bar{y}_{q + s_1 - s_0}^*(s_0) \Vert^2 \nonumber\\
		&~~~+ 2\Vert \bar{y}_q(s_1) - \bar{y}_{q + s_1 - s_0}^*(s_0) \Vert (1 +  V_{in})\big)\nonumber\\
		&\le \overline{\lambda}_{R_2} (\beta_1^2(\bar{v}, q+ s_1 - s_0) + 2 \beta_1(\bar{v}, q+ s_1 - s_0)(1 +  V_{in}))\nonumber\\
		&\stackrel{\triangle}{=} \beta_3(\bar{v},q)\label{eq:beta_3}.
		\end{align}
		Therefore, substituting \eqref{eq:l1} and \eqref{eq:beta_3} into \eqref{eq:ll}, we one has that
		\begin{align}
		& \sum_{q = 0}^{L - 1}\ell(\bar{u}_q(s_1), \bar{y}_q(s_1)) - \sum_{q = 0}^{L - 1}\ell(\bar{u}_q^*(s_0), \bar{y}_q^*(s_0)) \\
		&\le- \!\!\!\!\sum_{q = 0}^{s_1 - s_0 - 1}\!\!\!\!\ell(\bar{u}_q(s_0), \bar{y}_q(s_0)) + \overline{\lambda}_{\{R_1, R_2\}}\gamma_{uy}^2\Vert \Theta_\eta^{\dag} \Vert^2 \alpha_4(\bar{v})\nonumber\\
		&~~~ +\!\!\!\!\!\!\!\!\!\!\!\!\! \sum_{q = 0}^{L - \eta - (s_1 - s_0) - 1}\!\!\!\!\!\!\!\!\!\!\!\! \beta_3(\bar{v},q)\\
		& = - \!\!\!\!\sum_{q = 0}^{s_1 - s_0 - 1}\!\!\!\!\ell(\bar{u}_q(s_0), \bar{y}_q(s_0)) + \alpha_5(\bar{v})\label{eq:lls1}
		\end{align}
		where $\alpha_5(\bar{v}) := \overline{\lambda}_{\{R_1, R_2\}}\gamma_{uy}^2\Vert \Theta_\eta^{\dag} \Vert^2 \alpha_4(\bar{v}) + \sum_{q = 0}^{L - \eta - (s_1 - s_0) - 1} \beta_3(\bar{v},q)$ is a $\mathcal{K}_{\infty}$-function.
		
		Now, it remains to derive upper bounds for $\Vert \bar{g}(s_1)\Vert$ and $\Vert \bar{h}(s_1) \Vert$.
		From \eqref{eq:gbars1}, one gets that
		\begin{align}
		\Vert \bar{g}(s_1)\Vert^2 &= \Vert H_{ux}^{\ddag}\Vert^2\big(\Vert x_{s_1 - \eta}\Vert^2 + \Vert \bar{u}_{[s_1 - \eta, s_1 + L - 1]}\Vert^2\big)\\
		& = \Vert H_{ux}^{\ddag}\Vert^2\big(\Vert x_{s_1 - \eta}\Vert^2 \!+\! \Vert \bar{u}^*_{[s_1 - s_0  - \eta, L - \eta - 1](s_0)}
		\Vert^2 \nonumber\\
		&~~~+ \Vert \bar{u}_{[L - \eta - (s_1 - s_0), L - 1]}(s_1)\Vert^2\big)\label{eq:gbars121}\\
		& \le \Vert H_{ux}^{\ddag}\Vert^2\big(\Vert x_{s_1 - \eta}\Vert^2 \!+\! \Vert \bar{u}^*_{[s_1 - s_0  - \eta, L - \eta - 1](s_0)}
		\Vert^2\big) \nonumber\\
		&~~~+ \Vert H_{ux}^{\ddag}\Vert^2\gamma_{uy}^2\Vert x_{s_0 + L - \eta}\Vert^2\label{eq:gbars122}\\
		& \le \Vert H_{ux}^{\ddag}\Vert^2\big(\Vert x_{s_1 - \eta}\Vert^2 \!+\! \Vert \bar{u}^*_{[s_1 - s_0  - \eta, L - \eta - 1]}(s_0)
		\Vert^2\big) \nonumber\\
		&~~~+ \Vert H_{ux}^{\ddag}\Vert^2\gamma_{uy}^2\Vert \Theta_{\eta}\Vert^2\alpha_4^2(\bar{v})\label{eq:gbars123}
		\end{align}
		where \eqref{eq:gbars122} holds true because of \eqref{eq:gammauy2}, and \eqref{eq:gbars123} is obtained by substituting \eqref{eq:xyhat} and \eqref{eq:alpha4} into \eqref{eq:gbars122}.
		Based on \eqref{eq:hbar1}--\eqref{eq:hbar3}, and noticing that $s_1 - s_0 \le T$ from Lemma \ref{lem:dos}, one gets that
		\begin{align*}
		&(\lambda_h/\bar{v})\Vert \bar{h}(s_1)\Vert^2 \le(\lambda_h/\bar{v})\Bigg[ 2 \Vert \Upsilon_\eta\Vert^2\eta\bar{v}^2 + 2 \eta\bar{v}^2 \\
		&+ 2\Bigg(\!2\Bigg(b_1\!\!\! \!\sum_{j = L + \eta}^{N -1}\!\!\!\!\xi^j\Bigg)^2 (L + \eta)(N - L - \eta)\bar{v}^2 \nonumber\\
		&+ 2 \Vert\Upsilon_{L + \eta} \Vert^2(L + \eta)(N - L - \eta + 1)\bar{v}^2\Bigg)\Vert \bar{g}(s_1)\Vert^2 \\
		&+ 2\bigg\Vert\Theta_{T}\!\!\!\! \sum_{q = 0}^{L - n_x - 1}\!\!\!\!A^q\bigg\Vert^2\bar{v}^2 + 2\Vert\Upsilon_{L - \eta - n_x} \Vert^2(L - \eta - n_x)\bar{v}^2\Bigg],
		\end{align*}
		together with \eqref{eq:gbars123} yielding 
		\begin{align}
		&(\lambda_g\bar{v})\bar{g}(s_1) + (\lambda_h/\bar{v})\bar{h}(s_1)\le \bar{v}(\lambda_g + \gamma_4\lambda_h) \big(\Vert x_{s_1 - \eta}\Vert^2 \nonumber\\
		&+ \Vert \bar{u}^*_{[s_1 - s_0  - \eta, L - \eta - 1]}(s_0)
		\Vert^2\big) +\alpha_6(\bar{v})\label{eq:hbars1}
		\end{align}
		where $\gamma_4 := 2[2(b_1\sum_{j = L + \eta}^{N -1}\xi^j)^2 (L + \eta)(N - L - \eta) + 2 \Vert\Upsilon_{L + \eta} \Vert^2(L + \eta)(N - L - \eta + 1)]\Vert H_{ux}^{\ddag}\Vert^2$ and $\alpha_6$ is a $\mathcal{K}_{\infty}$-function.
		Combining \eqref{eq:lls1}, \eqref{eq:gbars123} and \eqref{eq:hbars1}, it can be obtained from \eqref{eq:ls1ls0} that
		\begin{align}
		&J^*_L(\tilde{z}(s_1)) - J^*_L(\tilde{z}(s_0)) \le- \!\!\!\!\!\sum_{q = 0}^{s_1 - s_0 - 1}\!\!\!\!\!\ell(\bar{u}_q(s_0), \bar{y}_q(s_0))+ \alpha_5(\bar{v})\nonumber\\
		& \! + \!\bar{v}(\lambda_g \!+\! \gamma_4\lambda_h) \big(\Vert x_{s_1 - \eta}\Vert^2 \!+\! \Vert \bar{u}^*_{[s_1 - s_0  - \eta, L - \eta - 1]}\!(s_0)
		\!\Vert^2\big) \!+\!\alpha_6(\bar{v})\nonumber.
		\end{align}
		Moreover, it follows from \eqref{eq:uiossV} that
		\begin{align*}
		&W(z_{s_1}) - W(z_{s_0}) \\
		&= \big(W(z_{s_1}) - W(z_{s_1 - 1})\big) + \cdots + \big(W(z_{s_0 + 1}) - W(z_{s_0})\big)\\
		&\le -\sigma_1\Vert z_{[s_0, s_1 - 1]}\Vert^2 + \sigma_2\Vert u_{[s_0, s_1 - 1]}\Vert^2 + \sigma_3\Vert y_{[s_0, s_1 - 1]}\Vert^2\\
		& \le -\sigma_1\Vert z_{[s_0, s_1 - 1]}\Vert^2 + \sigma_3(2\Vert\bar{y}^*_{[0,s_1 - s_0 - 1]}(s_0)\Vert^2\\
		&~~~ + 2\Vert y_{[s_0, s_1 - 1]} - \bar{y}^*_{[0,s_1 - s_0 - 1]}(s_0)\Vert^2)+ \sigma_2\Vert u_{[s_0, s_1 - 1]}\Vert^2 \\
		&\stackrel{\eqref{eq:beta_1}}{\le}\!\! -\sigma_1\Vert z_{[s_0, s_1 - 1]}\Vert^2 + \sigma_2\Vert u_{[s_0, s_1 - 1]}\Vert^2\\
		&~~~ + 2\sigma_3\Vert\bar{y}^*_{[0,s_1 - s_0 - 1]}(s_0)\Vert^2+ 2\sigma_3\sum_{q = 0}^{s_1 - 1}\beta_1^2(\bar{v}, q).
		\end{align*}
		Therefore, the Lyapunov function satisfies
		\begin{align*}
		&V_{s_1} \!-\! V_{s_0} = J^*_L(\tilde{z}(s_1)) - J^*_L(\tilde{z}(s_0)) + \gamma(W(z_{s_1}) - W(z_{s_0})) \\
		&\le - \!\!\!\!\!\sum_{q = 0}^{s_1 - s_0 - 1}\!\!\!\!\!\ell(\bar{u}_q(s_0), \bar{y}_q(s_0))+ \alpha_5(\bar{v})\!+\!\alpha_6(\bar{v})\nonumber\\
		& ~~~ + \bar{v}(\lambda_g \!+\! \gamma_4\lambda_h) \big(\Vert x_{s_1 - \eta}\Vert^2 \!+\! \Vert \bar{u}^*_{[s_1 - s_0  - \eta, L - \eta - 1]}\!(s_0)
		\!\Vert^2\big)  \\
		&~~~+ \gamma\big(-\sigma_1\Vert z_{[s_0, s_1 - 1]}\Vert^2 + \sigma_2\Vert u_{[s_0, s_1 - 1]}\Vert^2\\
		&~~~ + 2\sigma_3\Vert\bar{y}^*_{[0,s_1 - s_0 - 1]}(s_0)\Vert^2+ 2\sigma_3\sum_{q = 0}^{s_1 - 1}\beta_1^2(\bar{v}, q)\big)\\
		& \stackrel{\gamma = \frac{\overline{\lambda}{\{R_1, R_2\}}}{\sigma_2, 2\sigma_3}}{\le}\!\!\!\!\!\! \!\!\!\!\!\!  \underbrace{\alpha_5(\bar{v})\!+\!\alpha_6(\bar{v})\!+\! 2\gamma\sigma_3\sum_{q = 0}^{s_1 - 1}\beta_1^2(\bar{v}, q)}_{\stackrel{\triangle}{=}\alpha_7(\bar{v}) \in \mathcal{K}_{\infty}}\!-\Big(\frac{\gamma\sigma_1}{\gamma_z}\Big)\Vert x_{s_0}\Vert^2\nonumber\\
		& ~~~ + \big(\bar{v}(\lambda_g \!+\! \gamma_4\lambda_h) - (\gamma\sigma_1/\gamma_z)\big)\Vert x_{s_1 - \eta}\Vert^2 \\
		&~~~+\bar{v}(\lambda_g + \gamma_4\lambda_h) \Vert \bar{u}^*_{[s_1 - s_0  - \eta, L - \eta - 1]}\!(s_0)
		\!\Vert^2  \\
		&\le  \alpha_7(\bar{v})\! -\! (\gamma\sigma_1/\gamma_z)\Vert x_{s_0}\Vert^2 \!+\! \big(\bar{v}(\lambda_g \!+\! \gamma_4\lambda_h) \!-\! (\gamma\sigma_1/\gamma_z)\big)\\
		&~~~\times\bigg\Vert A^{s_1 - \eta - s_0} x_{s_0} + \!\!\!\!\!\!\!\!\!\!\sum_{q = 0}^{s_1 - s_0 - \eta - 1}\!\!\!\!\!\!\!\!\!\!A^qBu_{s_1 - \eta - q} +\!\!\!\!\!\!\!\!\!\!\sum_{q = 0}^{s_1 - \eta - s_0 - q}\!\!\!\!\!\!\!\!\!\!w_{s_1 - \eta - q}\bigg\Vert^2\\
		&~~~+\bar{v}(\lambda_g + \gamma_4\lambda_h) \Vert \bar{u}^*_{[s_1 - s_0  - \eta, L - \eta - 1]}\!(s_0)
		\!\Vert^2 \\
		&\stackrel{\bar{u} \in \mathbb{U}}{\le}\!\!\!\! -\Big(\frac{\gamma \sigma_1}{\gamma_z} \!-\! 2\big(\bar{v}(\lambda_g \!+\! \gamma_4\lambda_h) \!-\! \frac{\gamma \sigma_1}{\gamma_z}\big)\Vert A^{T - \eta}\Vert\!\Big)\Vert x_{s_0}\Vert^2 \!+\! \underbrace{\alpha_7(\bar{v})}\\
		&~~~ \underbrace{+ \bar{v}\bigg(\!\!4(\!\lambda_g \!+\! \gamma_4\lambda_h \!-\! \frac{\gamma \sigma_1}{\gamma_z}\!)\!\!\!\!\!\sum_{q = 0}^{T - \eta - 1}\!\!\!\!\!\Vert A^qB\Vert^2 \!+\! (\lambda_g \!+\! \gamma_4\lambda_h)\!\!\bigg)u_{\max}^2}_{\stackrel{\triangle}{=}\alpha_8(\bar{v}) \in \mathcal{K}_{\infty}}\\
		&\le -\gamma_5\Vert x_{s_0}\Vert^2 + \alpha_8(\bar{v})
		\end{align*}
		where $\gamma_5\! :=\!(\gamma \sigma_1/\gamma_z) \!-\! 2\big(\bar{v}(\lambda_g \!+\! \gamma_4\lambda_h) - (\gamma \sigma_1/\gamma_z)\big)\Vert A^{T - \eta}\Vert > 0$ holds true if $\bar{v}_0 \le ((2\Vert A^{T - \eta}\Vert + 1)\gamma \sigma_1/(2\Vert A^{T - \eta}\Vert(\lambda_g + \gamma_4\lambda_h)\gamma_z))$.
		Further, adopting \eqref{eq:boundV} in Lemma \ref{lem:boundV}, and choosing $\gamma_x >0$ such that $\gamma_x\Vert z_t\Vert \le \Vert x_t\Vert$, one gets that
		\begin{equation}\label{eq:vs1-vs0}
		V_{s_1} \!-\! V_{s_0} \le \!-(\gamma_5\gamma_x^2/\gamma_3)V_{s_0} \!+\! (\gamma_5\gamma_x^2/\gamma_3)\alpha_3(\bar{v}) \!+ \!\alpha_8(\bar{v})
		\end{equation}
		In conclusion, the Lyapunov function decreases at every successful transmission time instant, and the Problem \eqref{eq:mpc} is feasible at $s_r$ for all $r \in \mathbb{N}_0$.
		
		{\emph{s2) ISS}}
		
		Based on \eqref{eq:vs1-vs0}, it can be obtained recursively that
		\begin{equation*}
		V_{s_{r}} \!\le\! \bigg(\!\!1 - \frac{\gamma_5\gamma_x^2}{\gamma_3}\!\bigg)^{r + 1} \!\!\!\!\!\!\!\!\!V_{s_0} \!+\! \sum_{q = 0}^{r}\bigg(\!1 \!-\! \frac{\gamma_5\gamma_x^2}{\gamma_3}\!\bigg)^{q}\!\!\bigg(\!\frac{\gamma_5\gamma_x^2}{\gamma_3}\alpha_3(\bar{v}) \!+ \!\alpha_8(\bar{v})\!\!\bigg).
		\end{equation*}
		Finally, adopting \eqref{eq:boundV}, we arrive at
		\begin{align}
		\Vert z_{s_r} \Vert^2 \le \frac{\gamma_3}{\gamma\underline{\lambda}_{P}}\bigg(\!\!1 - \frac{\gamma_5\gamma_x^2}{\gamma_3}\!\bigg)^{r + 1} \!\!\!\!\!\!\!\!\!\Vert z_{s_0} \Vert^2 + \alpha_9(\bar{v})
		\end{align}
		with function $\alpha_9(\bar{v}) \in \mathcal{K}_{\infty}$, which completes the proof according to Definition \ref{def:iss}.
	\end{proof}
	\begin{remark}[{\emph{Global ISS}}]
		In Theorem \ref{thm:isps}, $z_0 \in \mathbb{B}_{\delta}$ is required to establish the recursive feasibility. 
		That is, system \eqref{eq:sys} achieves only local ISS under the proposed data-driven resilient controller in Algorithm \ref{alg:mpc}.
		To recover global ISS, we provide next two solutions, each of which comes at the price of either reduced resilience against DoS attacks or increased computational complexity.
		The first solution is to solve Problem \eqref{eq:mpc} once every $n_x$ time instants, and apply the control inputs in $u_{[t, t + n_x - 1]} = \bar{u}^*_{[0, n_x - 1]}(t)$ sequentially to the system over the ensuing $n_x$ steps.
		Under such a method, global recursive feasibility of Problem \eqref{eq:mpc} can be shown similarly to that of Theorem \ref{thm:isps}, and the closed-loop system achieves global ISS provided that condition $(1/\nu_f) + (1/\nu_d) < 1 - ((n_x - 1)/\nu_f)$ on DoS attacks is met.
		In addition, a data-driven MPC scheme without the terminal constraint \eqref{eq:mpc3} was recently proposed in \cite{berberich2021robust}, under which global ISS can be ensured if the prediction horizon  (i.e., $L$) is sufficiently large.
		Condition on $L$ implies that this scheme also requires sufficiently many pre-collected trajectories (i.e., large $N$).
		In this manner, system \eqref{eq:sys} can achieve maximum resilience against DoS attacks at the price of increased computational complexity.
		Stability analysis for these two schemes can be performed by leveraging the results in \cite{berberich2019data} and \cite{berberich2021robust}.
	\end{remark}
	\begin{remark}[{\emph{DoS attacks also in controller-to-plant channel}}]
		When both controller-to-plant and sensor-to-controller channels are subject to DoS attacks, a packetized transmission policy can be used in the input channel too.
		At each successful transmission time instant, the plant receives a packet consisting of $b$ inputs from the controller (i.e., $\bar{u}^*_{0,b - 1}(s_r)$), where $b$ is the buffer size at the plant side.
		Our proposed data-driven resilient controller along with the stability and robustness guarantees of this paper can be generalized to this case, and the trade-off between the buffer size $b$ and DoS attacks can be characterized by the decrease rate and the increase rate of the Lyapunov function (c.f. \cite{FengNetworked}).
	\end{remark}
	
	\section{Numerical Example}\label{sec:example}
	Consider the unstable batch reactor example in \cite{FengResilient} (originally studied in \cite{Walsh2002Scheduling}), where model-based resilient control was investigated under DoS attacks.
	This model is given by $\dot{x}(t) = A x(t) + B u(t)$ and $y = C x(t)$, with
	\begin{align*}
	&A :=
	\left[
	\begin{matrix}
	1.38 & -0.2077 & 6.715 & -5.676\\
	-0.5814 & -4.29 & 0 & 0.675\\
	1.067 & 4.273 & -6.654 & 5.893\\
	0.048 & 4.273 & 1.343 & -2.104
	\end{matrix}
	\right]\\
	&B := \left[
	\begin{matrix}
	0 & 0\\
	5.679 & 0\\
	1.136 & -3.146\\
	1.136 & 0
	\end{matrix}
	\right],~~~~
	C :=
	\left[
	\begin{matrix}
	1 & 0 & 1 & -1\\
	0 & 1 & 0 & 0
	\end{matrix}
	\right].
	\end{align*}
	
	Here, we consider a discrete-time version of the unstable batch reactor with a sampling period of $0.1s$, and evaluate the resilient performance of the proposed controller.
		To this end, a number of input-output trajectories of length $N = 100$ were obtained by means of simulating the open-loop system off-line using input sequences obeying Assumption \ref{as:multitre}.
		Based on Assumption \ref{as:horizon}, parameters of Algorithm \ref{alg:mpc} were set as follows: $L = 10 \ge n_x + 2\eta =  8$, $\lambda_g = 0.1$, $\lambda_h = 100$, $R_1 = 10^{-4} I_{2}$, and $R_2 = 3{\rm I_{2}}$.
		Over a simulation horizon of $200$ time instants, Figs. \ref{fig:y7}--\ref{fig:y85} depict the comparative performance of the  data-driven method and the model-based method under different levels of noise and DoS attacks (signified by the gray shades). 
		In addition, the identification step uses the same data  as the data-driven method.
		The red dashed line (output using the proposed controller) in Figs. \ref{fig:y7}--\ref{fig:y85} confirmed that there is a trade-off between system resilience against DoS attacks and robustness against noise, which aligns with our observations in Remark \ref{rmk:tradeoff}.
		Compared with the data-riven method, the model-based method either using the controller in \cite{FengResilient} or standard model predictive controller yields poorer system performance, which resembles the results in \cite{krishnan2021On}.
		This is because the identification step tend to over-fit the noise in the pre-collected data when system dimension is high (e.g., $n_x \ge 4$).
		
		Furthermore, relationship between system performance, prediction horizon $L$, and length of the pre-collected data $N$ is discussed in Fig. \ref{fig:ynormNL}.
		It was shown in \cite{persis2020data} that the length of pre-collected data should obey $N \ge (\max\{n_u, n_y\} + 1)L + n_x -1 = 33$.
		The top panel of Fig. \ref{fig:ynormNL} illustrates that when $N$ is small (i.e., $33< N \le 60$), system performance improves fast as $N$ increases; and when $N$ becomes large, system performance remains almost steady as $N$ increases.
		The bottom panel of Fig. \ref{fig:ynormNL} reports system performance for different $L$ values.
		According to Assumption \ref{as:horizon}, the predicted horizon obeys $2\eta + n_x \le L \le (N - n_x + 1)/(\max\{n_u, n_y\} + 1)$.
		Therefore, when $N = 40$, one gets that $8 \le L \le 12$.
		It has been shown that, as long as the $L$ and $N$ are enough for evaluating Problem \eqref{eq:mpc}, they have little effect on the system performance.
		Finally, the bottom panel of Fig. \ref{fig:ynormNL} implies that when the dataset is small, the data-driven control method achieves better performance than using the model-based one.
	\begin{figure}
		\centering
		\includegraphics[width=9cm]{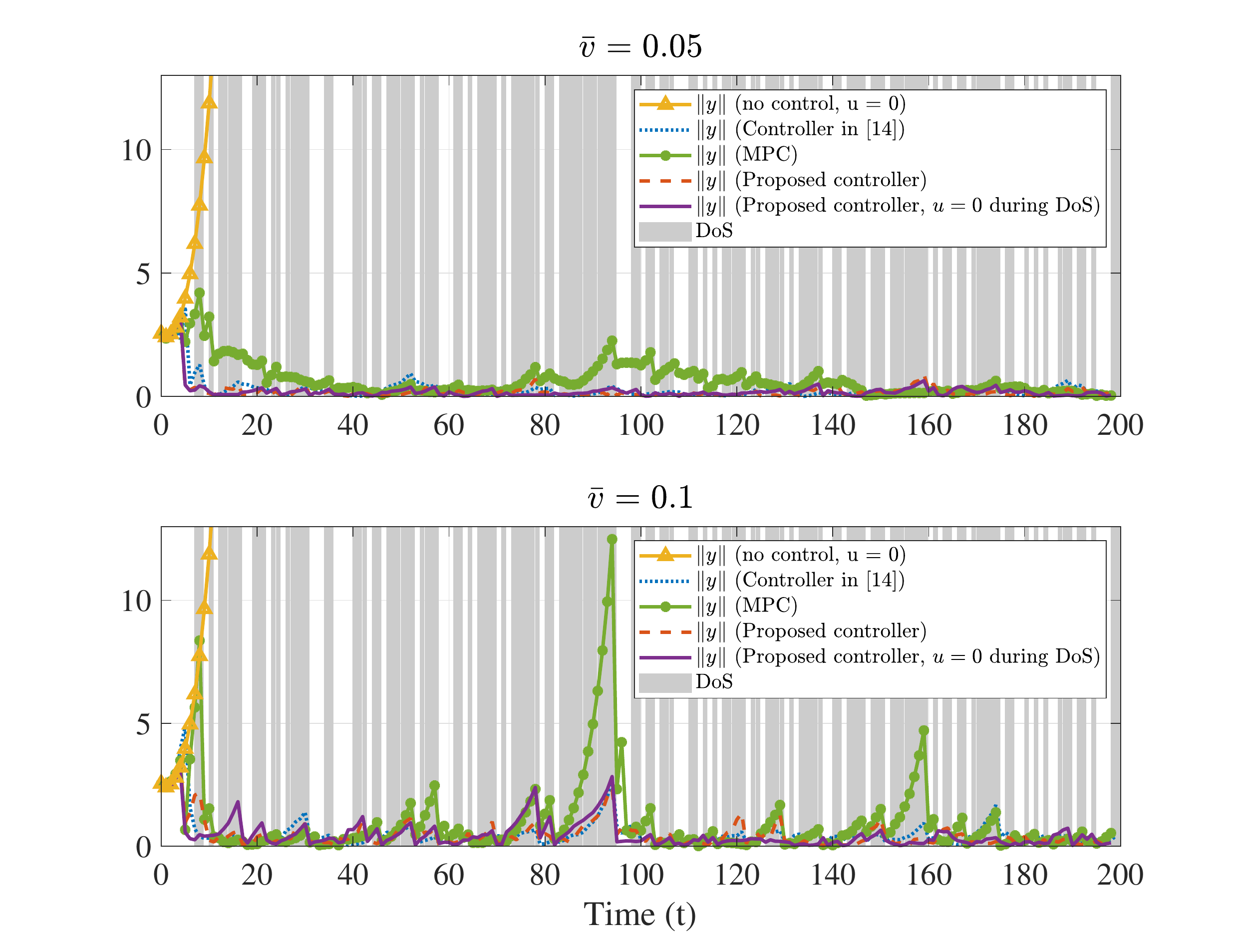}\\
		\caption{Norm of the closed-loop output $\Vert y\Vert$: $1/\nu_f + 1/\nu_d = 0.8841$.}\label{fig:y7}
		\centering
	\end{figure}
	\begin{figure}
		\centering
		\includegraphics[width=9cm]{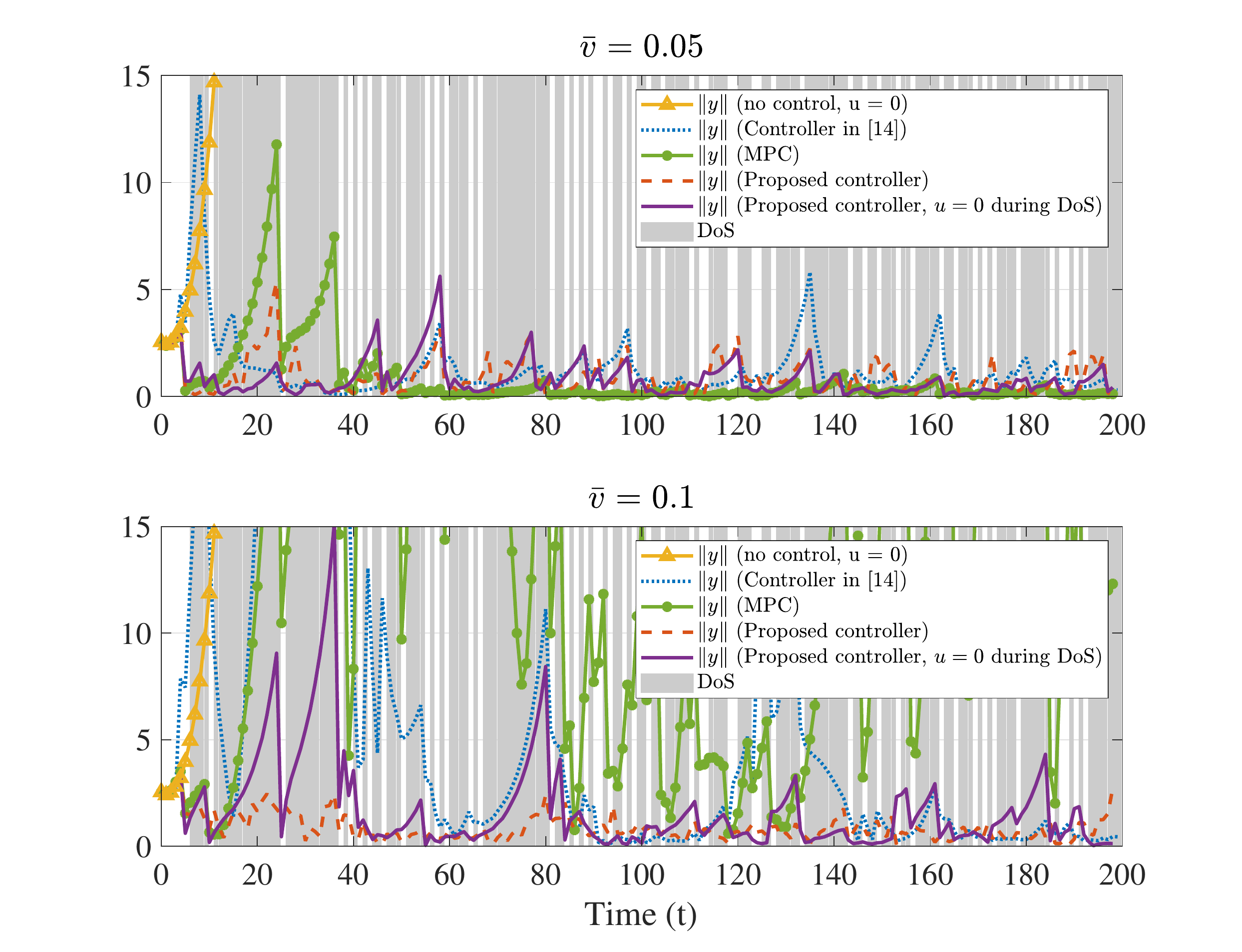}\\
		\caption{Norm of the closed-loop output $\Vert y\Vert$: $1/\nu_f + 1/\nu_d = 0.9142$.}\label{fig:y8}
		\centering
	\end{figure}
	\begin{figure}
		\centering
		\includegraphics[width=9cm]{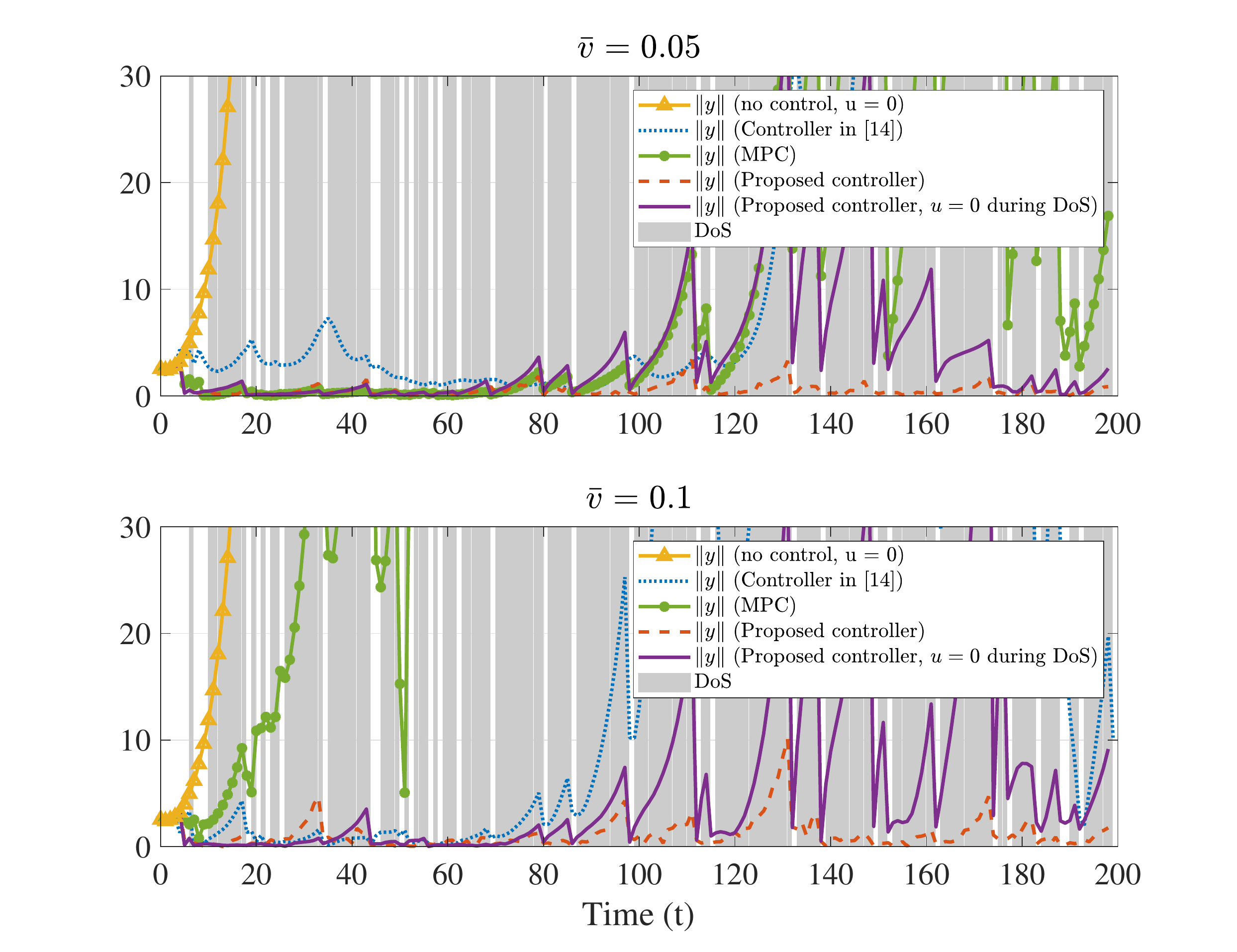}\\
		\caption{Norm of the closed-loop output $\Vert y\Vert$: $1/\nu_f + 1/\nu_d = 0.9317$.}\label{fig:y85}
		\centering
	\end{figure}
	\begin{figure}
		\centering
		\includegraphics[width=9cm]{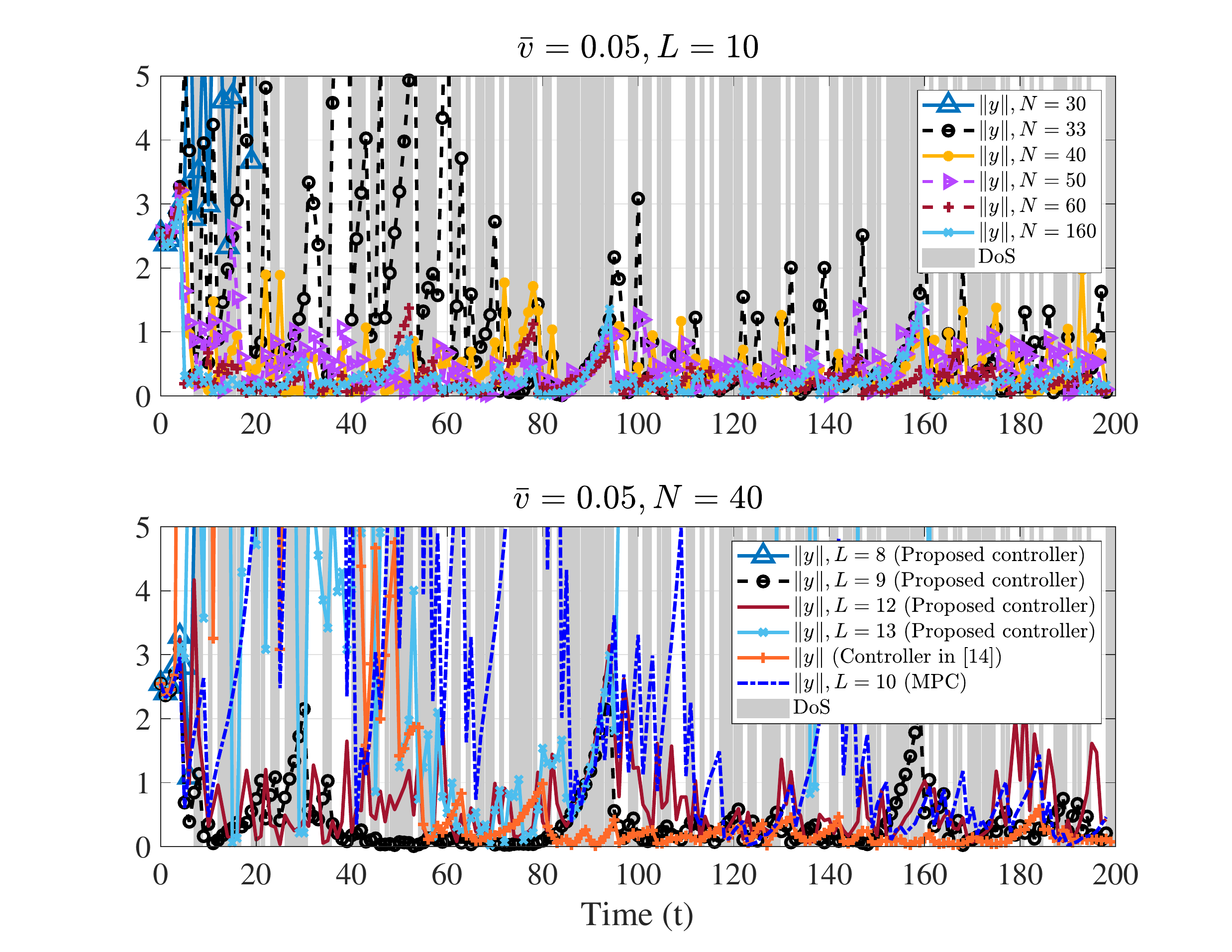}\\
		\caption{Relationship between system performance, $L$, and $N$: $1/\nu_f + 1/\nu_d = 0.8841$.}\label{fig:ynormNL}
		\centering
	\end{figure}
	
	\section{Conclusions}\label{sec:conclusion}
	This paper considered the stabilization problem of unknown stochastic LTI systems under DoS attacks. 
	A data-driven MPC scheme was developed, under which a sequence of inputs and associated predicted outputs can be computed purely from some pre-collected input-output data by solving convex programs. 
	On top of this scheme, a data-driven resilient controller was proposed such that the system can achieve local input-to-state stability under conditions on the DoS attacks and level of noise, even without any knowledge of system model or a priori system identification procedure.
	Moreover, the condition on DoS attacks is the same as the maximum resilience achieved in existing works under model-based controller.
	To achieve global stability, two modifications of the proposed controller were discussed at the price of sacrificing system resilience against DoS attacks, or increasing computational complexity. 
	Finally, a numerical example was provided to demonstrate the effectiveness of the proposed method as well as the practical merits of the established theory.
	\bibliographystyle{IEEEtran}
	
	\bibliography{bible3}

\end{document}